\documentclass[journal,twocolumn,letterpaper]{IEEEtran}
\usepackage{ifpdf}
\usepackage{cite}

\usepackage{color}
\usepackage{amsthm}
\usepackage{amsmath,amssymb}
\usepackage{algorithmic}
\usepackage{array}
\usepackage{mdwmath}
\usepackage{mdwtab}
\usepackage{graphicx}
\usepackage{lipsum}
\usepackage{epstopdf }
\usepackage{algorithm,algorithmic}

\usepackage{enumerate}
\usepackage{eqparbox}
\usepackage{url}

\usepackage{subfig}
\usepackage{amsbsy}

\usepackage{mathtools}

\def\cJ{\mathcal{J}}

\def\cN{\mathcal{N}}

\newtheorem{theorem}{Theorem}
\newtheorem{proposition}{Proposition}

\begin{document}
\title{\huge An Online Optimization Framework for Distributed Fog Network Formation with Minimal Latency }
\author{\IEEEauthorblockN{Gilsoo~Lee$^{*}$,~Walid~Saad$^{*}$, and~Mehdi~Bennis$^{\dag}$}\\
\IEEEauthorblockA{ \small $^{*}$ Wireless@VT, Department of Electrical and Computer Engineering, Virginia Tech, Blacksburg, VA, USA, \\
Emails: \protect\url{{gilsoolee, walids}@vt.edu}. \\
\small $^\dag$ Centre for Wireless Communications, University of Oulu, Finland, 
 Email: \url{bennis@ee.oulu.fi}. \vspace{-7mm}\\
}
\thanks{A preliminary conference version \cite{lee2017online} of this work was presented at  IEEE ICC 2017.}
}

\maketitle

\begin{abstract} 
Fog computing is emerging as a promising paradigm to perform distributed, low-latency computation by jointly exploiting the radio and computing resources of end-user devices and  cloud servers. However, the dynamic and distributed formation of local fog networks is highly challenging due to the unpredictable arrival and departure of  neighboring fog nodes. Therefore, a given fog node must properly select a set of neighboring nodes and intelligently offload its computational tasks to this set of neighboring fog nodes and the cloud in order to achieve low-latency transmission and computation. In this paper, the  problem of fog network formation and  task distribution  is jointly investigated while considering a hybrid fog-cloud architecture. The overarching goal is to minimize the maximum communication and computation latency by enabling a given fog node to form a suitable fog network and optimize the task distribution,  under  uncertainty on the arrival process of neighboring fog nodes. To solve this problem, a novel online optimization framework is proposed in which  the neighboring nodes are selected by using a threshold-based online algorithm that uses a target competitive ratio, defined as the ratio between the latency of the online algorithm and the offline optimal latency. The proposed framework repeatedly updates its target competitive ratio and optimizes the distribution of the fog node's computational tasks in order to minimize latency. Simulation results show that, for specific settings, the proposed framework can successfully select a set of neighboring nodes while reducing  latency by up to $19.25$\% compared to a baseline approach based on the well-known online secretary framework. The results also show how, using the proposed framework, the computational tasks can be properly offloaded between the fog network and a remote cloud server in different network settings. 
\end{abstract}


\IEEEpeerreviewmaketitle

\vspace{-3mm}
\section{Introduction}
The Internet of Things (IoT) is expected to connect over 50 billion things worldwide, by 2020 \cite{dawy2017toward, mozaffari2016unmanned,  park2016learning}. To meet the low-latency requirement of task computation for the IoT devices, relying on conventional, remote cloud solutions may not be suitable due to the high end-to-end transmission latency of the cloud \cite{chiang2016fog}. Therefore, to reduce the transmission latency, the local proximity of IoT devices can be exploited for offloading computational tasks, in a distributed manner. Such local computational offload gives rise to the emerging paradigm of \emph{fog computing}~\cite{cisco2015fog}. Fog computing also known as edge computing allows overcoming the limitations of centralized cloud computation by enabling distributed, low-latency computation at the network edge, for supporting various wireless and IoT applications \cite{peng2016fog}. The advantages of the fog architecture comes from the transfer of some of the network functions to the network edge. Indeed, significant amounts of data can be stored, controlled, and computed over fog networks that can be configured and managed by end-user nodes \cite{chiang2016fog}. Within the fog paradigm, computational tasks can be intelligently allocated between the fog nodes and the cloud to meet computational and latency requirements \cite{elbamby2017proactive}. To implement  the fog paradigm, a three-layer network architecture is typically needed to manage sensor, fog, and cloud layers \cite{peng2016fog}. When the computing tasks are offloaded from the sensor layer to the fog and cloud layers, fog computing faces a number of challenges such as fog network formation and radio/computing resource allocation \cite{lee2017mmtc}. In particular, it is challenging for fog nodes to dynamically form and maintain a fog network that they can use for offloading their task. This challenge is exacerbated by the fact that fog computing devices are inherently mobile and will join/leave a network sporadically \cite{yannuzzi2014key}. Moreover, to efficiently use the computing resource pool of the fog network, novel resource management schemes for the hybrid fog-cloud network architecture are needed~\cite{bonomi2012fog}. 

To reap the benefits of fog networks, many architectural and operational challenges must be addressed~\cite{vallati2015exploiting, khelil2014suitability, luan2011towards, nishio2013service, sharma2017saca, zhao2017tasks, kaewpuang2013framework, khaledi2016profitable, ketyko2016multi, souza2016handling, park2016joint, yu2016joint, deng2015towards, mao2016power}. A number of approaches for fog network formation are investigated in \cite{vallati2015exploiting, khelil2014suitability, luan2011towards, nishio2013service, sharma2017saca}. To configure a fog network, the authors in \cite{vallati2015exploiting} propose the use of a device-to-device (D2D)-based network  that can efficiently support networking between a fog node and a group of sensors. Also, to enable connectivity for fog computing, the work in \cite{khelil2014suitability} reviews D2D techniques that can be used for reliable wireless communications among highly mobile nodes. The work in \cite{luan2011towards} proposes a framework for vehicular fog computing in which fog servers can form a distributed vehicular network for content distribution. In \cite{nishio2013service}, the authors study a message exchange procedure to form a local network for resource sharing between the neighboring fog nodes. The work in \cite{sharma2017saca} introduces a method to form a hybrid fog architecture in the context of transportation and drone-based networks. 

Once a fog network is formed, the next step is to share resources and tasks among fog nodes as studied in \cite{zhao2017tasks, kaewpuang2013framework, khaledi2016profitable, ketyko2016multi, souza2016handling, park2016joint, yu2016joint, deng2015towards, mao2016power}. For instance, the work in \cite{zhao2017tasks} investigates the problem of scheduling tasks over heterogeneous cloud servers in different scenarios in which multiple users can offload their tasks to the cloud and fog layers. The work in \cite{kaewpuang2013framework} studies the joint optimization of radio and computing resources using a game-theoretic approach in which mobile cloud service providers can decide  to cooperate in resource pooling. Meanwhile, in \cite{khaledi2016profitable}, the authors propose a task allocation approach that minimizes  the overall task completion time by  using a multidimensional auction and finding the best time interval between multiple auctions to reduce unnecessary time overheads. The authors in \cite{ketyko2016multi} study a latency minimization problem to allocate the computational resources of the mobile-edge servers. Moreover, the authors in \cite{souza2016handling} study the delay minimization problem in  fog and cloud-assisted networks under heterogeneous delay considerations. Moreover, the work in \cite{park2016joint} investigates the problem of minimizing the aggregate cloud fronthaul and wireless transmission latency. In \cite{yu2016joint}, a task scheduling algorithm is proposed to jointly optimize the radio and computing resources to reduce the users' energy consumption while satisfying delay constraints. The problem of optimizing power consumption is also considered in \cite{deng2015towards} subject to delay constraint using a queueing-theoretic delay model  at the cloud. Moreover, the work in \cite{mao2016power} studies the power consumption minimization problem in an online scenario subject to uncertain task arrivals. Furthermore, the work in \cite{lee2017computational}, studies how tasks can be predicted and proactively scheduled. Last, but not least, the work in \cite{wang2017enorm} implements a prototype for fog computing that can manage edge node's resources in a distributed computing environment.

In all of these existing fog network formation and task scheduling works in fog networks \cite{luan2011towards, nishio2013service, sharma2017saca, zhao2017tasks, kaewpuang2013framework, khaledi2016profitable, ketyko2016multi, souza2016handling, park2016joint, yu2016joint, deng2015towards}, it is generally assumed that information on the formation of the fog network is completely known to all nodes. However, in practice, the fog network can be spontaneously initiated by a fog node when other neighboring fog nodes start to dynamically join or leave the network. Hence, the presence of a neighboring fog node to which one can offload tasks is  unpredictable. Indeed, it is challenging for a fog node  to know when and where another fog node will arrive. Thus, there exists an inherent \emph{uncertainty} stemming from the unknown locations  and availability of fog nodes. Further, most of the existing works \cite{luan2011towards, nishio2013service, khaledi2016profitable, ketyko2016multi, park2016joint, souza2016handling, yu2016joint } typically assume a simple transmission or computational latency model for a fog node. In contrast, the use of a queueing-theoretic model for both transmission and computational latency is necessary to capture realistic latency metrics. Consequently, unlike the existing literature \cite{nishio2013service, khaledi2016profitable, ketyko2016multi, souza2016handling, park2016joint, yu2016joint } which assumes full information knowledge for fog network formation and relies on simple delay models, our goal is to design an \emph{online approach} to enable an on-the-fly formation of the fog network, under uncertainty, while minimizing the computational latency given an end-to-end latency model.

The main contribution of this paper is a novel framework for online fog network formation and task distribution in a hybrid fog-cloud network. This framework allows any given fog node to dynamically  construct a fog network by selecting the most suitable set of neighboring fog nodes in presence of uncertainty on the arrival order of neighboring fog nodes. The fog node can jointly use its fog network as well as a distant cloud server to compute given tasks. We formulate an online optimization problem whose objective is to minimize the maximum computational latency of all fog nodes  by properly selecting the set of fog nodes to which computations will be offloaded while also properly distributing the tasks among those fog nodes and the cloud. To solve this problem without any prior information on the future arrival order of fog nodes, we propose an online optimization framework that achieves a target competitive ratio; defined as the ratio between the latency achieved by the proposed algorithm and the optimal latency that can be achieved by an offline algorithm. In the proposed framework, an online algorithm is used to form a fog network when the neighboring nodes arrive sequentially, the task distribution is optimized among the nodes on the formed network, and the target competitive ratio is repeatedly updated. We show the target competitive ratio can be achieved by iteratively running the proposed algorithm. Simulation results show that the proposed  framework can achieve a target competitive ratio of $1.21$  in a given simulation scenario. For a specific simulation setting, simulation results show that the proposed algorithm can reduce the latency by up to $19.25$\% compared to the baseline approach that is a modified version of the popular online secretary algorithm \cite{lee2017online}. Therefore, the proposed  framework is shown to be able to find a suitable competitive ratio that can reduce the latency of fog computing while properly selecting the neighboring fog nodes that have high performance and suitably distributing tasks across fog nodes and a cloud server. 

The rest of this paper is organized as follows. In Section~\ref{sec:systemmodel}, the system model is presented. We formulate the online problem in Section~\ref{sec:problemformulation}. In Section~\ref{sec:solution}, we propose our online optimization framework to solve the problem. In Section~\ref{sec:simulations}, simulation results are carried out to evaluate the performance of our proposed framework. Conclusions are drawn in Section~\ref{sec:conclusion}.

\vspace{-3mm}
\section{System Model}\label{sec:systemmodel}

\begin{figure}[]
\centering
\includegraphics[width=0.3\textwidth]{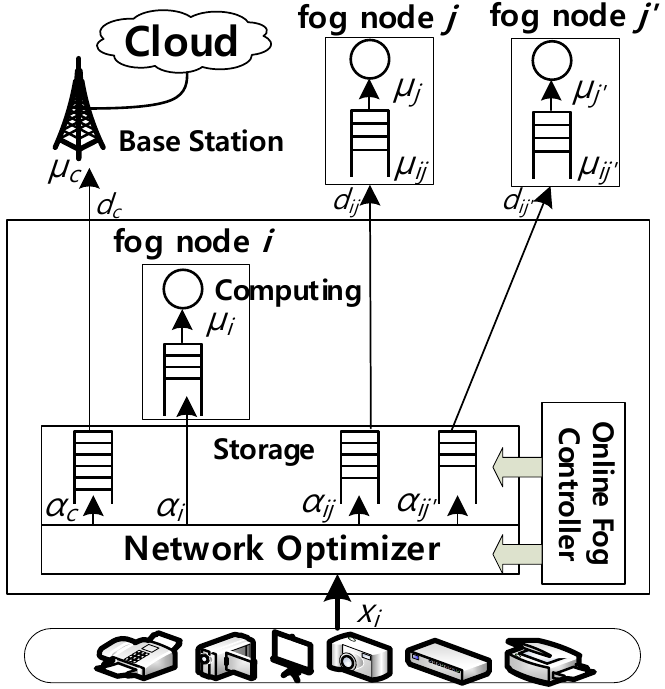}
\caption{\small System model of the fog networking architecture and the cloud.}
\label{fig:system}\vspace{-3mm}
\end{figure}

Consider a fog network  consisting of a sensor layer, a fog layer, and a cloud layer as shown in Fig.~\ref{fig:system}. In this system, the sensor layer includes smart and  small-sized IoT sensors with limited computational capability. Therefore, when sensors generate the computational tasks, the sensors' tasks are offloaded to the fog and cloud layers for purposes of remote distributed computing. Similarly, cloud tasks can also be offloaded to the fog layer. In our model, the cloud layer can be seen as the conventional cloud computing center. The fog layer refers to the set of IoT devices (also called fog nodes) that can perform fog computing jobs such as storing data and computing tasks. We assume that various kinds of sensors send their task data to a certain fog node $i$, and the data arrival rate to this node is $x_i$ packets per second where a task packet has a size of $K$ bits{\footnote{The initial fog node can gather data from any other node, including sensors or a cloud.}. Fog node $i$ performs the roles of collecting, storing, controlling, and processing the task data from the sensor layer, as is typical in practical fog networking scenarios~\cite{chiang2016fog}. In our architecture, for efficient computing, fog node $i$ must cooperate with other neighboring fog nodes and the cloud data center. We consider a network having a set $\mathcal{N}$ of $N$ fog nodes other than fog node $i$. For a given fog node $i$, we focus on the fog computing case in which fog node $i$ builds a network with a subset $\mathcal{J} \subset \mathcal{N}$  of $J$ neighboring fog nodes. Also, since the cloud is typically located at a remote location, fog node $i$ must access the cloud via wireless communication links using a cellular base station $c$.

Once the initial fog node $i$ receives tasks that arrive with  the rate of $x_i$ packets per second, it assigns a fraction of $x_i$ to other nodes. Then, each node within the considered fog-cloud network will locally compute the assigned fraction of $x_i$. The fraction of tasks locally computed by fog node $i$ is $\lambda_i(\alpha_i) = \alpha_i x_i$. Then, the task arrival rate offloaded from fog node $i$ to fog node $j\in \cJ$ is $\lambda_{ij}(\alpha_{ij}) = \alpha_{ij} x_i$. Therefore, the task arrival rate processed at the fog layer is $(\alpha_i + \sum_{j\in\cJ}\alpha_{ij}) x_i$. The number of remaining tasks $\lambda_c(\alpha_c) =\alpha_c x_i$ will then be offloaded to the cloud. When fog node~$i$ makes a decision on the distribution of all input tasks $x_i$, the task distribution variables are represented as vector ${\boldsymbol{\alpha}}=[\alpha_i, \alpha_c, \alpha_{i1}, \ldots, \alpha_{ij}, \ldots, \alpha_{iJ}]$ with $\sum_{j\in \cJ} \alpha_{ij} + \alpha_i + \alpha_c = 1$. Naturally, the total task arrival rate that arrives at fog node $i$ will be equal to the sum of the task arrival rates assigned to all computation nodes in the fog and cloud layers. Also, to model the random arrival of tasks from the sensors to fog node $i$, the total task arrival rate arriving at fog node $i$ can be modeled by a Poisson process~\cite{deng2015towards}. The tasks offloaded to the fog nodes and the cloud also follow a Poisson process if the tasks are randomly scheduled in a round robin fashion \cite{bertsekas1992data}. Also, the initial fog node can determine the transmission order of the task packets offloaded from the sensor layer. Therefore, in future work, if the tasks offloaded from the sensor layer have different service-level latency requirements, the initial fog node can prioritize urgent task packets in its queue. 

\begin{table}[]
	\centering
	\caption{Summary of notations}
	\label{table:variables}
	\begin{tabular}{|l|l|}
		\hline
		$i$ & Index of initial fog node                                   \\ \hline $j$ & Index of neighboring fog nodes in $\cJ$          \\ \hline
		$c$ & Index of cloud                                                \\ \hline $J=|\cJ|$ & Number of neighboring fog nodes          \\ \hline
		$x_i$ & {\color{black} Total task arrival rate from sensors to node $i$  }    \\ \hline $\alpha_{k\in \in \{i,ij,c\}}$ & Tasks offloaded toward $k$\\ \hline
		$\mu_{ij}$ & Fog transmission service rate from $i$ to $j$  \\ \hline $\mu_c$ & Cloud transmission service rate               \\ \hline
		$\mu_i$     & Computing service rate of fog node $i$  \\ \hline $\mu_j$ & Computing service rate of fog node $j$    \\ \hline
		$1/\omega_{k \in \{i,j,c\}}$&Processing speed of node $k$ \\ \hline  $n$ & Arrival order  \\ \hline
		$K$& Size of a task packet  \\ \hline $\gamma$ & Target competitive ratio \\ \hline
	\end{tabular}\vspace{-3mm}
\end{table}

When the tasks arrive from the sensors to fog node $i$, they are first saved in fog node $i$'s storage, incurring  a waiting delay before they are transmitted and distributed to other nodes (fog or cloud). This additional delay pertains to  the transmission from fog node $i$ to $c$ or $j$ and can be modeled using a \emph{transmission queue}. Moreover, when the tasks arrive at the destination, the latency required to perform the actual computations  will be captured by a \emph{computation queue}. In Fig.~\ref{fig:system}, we show examples of both queue types. For instance, for transmission queues, fog node~$i$  must maintain transmission queues for each fog node $j$ and the cloud $c$. For computation, each fog node has a computation queue. To model the transmission queue, the tasks are transmitted to fog node $j$ over a wireless channel. Then, the service rate (in packets per second) can be given by 
\begin{eqnarray}\label{eq:mu}
\mu_{ij} =\frac{W_l}{K}  \log_2\left(1+\frac{g_{ij} h P_{\textrm{tx},i}}{W_l N_0}\right),
\end{eqnarray}
where $g_{ij}$ is the channel gain between fog nodes $i$ and $j$ with $d_{ij}$ being the distance between them, and $h$ is the average fading gain of the fog node $i$. When the fog nodes are located in proximity within a similar environment, we assume that they have identical average fading gains. If $d_{ij} \leq 1~\textrm{m}$, $g_{ij} \triangleq \beta_1$, and, if $d_{ij} > 1~\textrm{m}$, $ g_{ij} \triangleq \beta_1 d_{ij}^{-\beta_2}$ where $\beta_1$ and $\beta_2$ are, respectively, the path loss constant and the path loss exponent. Also, $P_{\textrm{tx},i}$ is the transmission power of fog node $i$ and $N_0$ is the noise power spectral density. The bandwidth per node is given by $W_l$ where $l=1$ and $2$ indicate, respectively, two types of bandwidth allocation schemes: equal allocation and cloud-centric allocation.\footnote{The problem of joint bandwidth optimization and fog computing can be subject for future work.} For \emph{equal bandwidth allocation}, all nodes in the network will be assigned equal bandwidth, i.e., $W_1=\frac{B}{J+1}$ where the total bandwidth $B$ is equally shared by $J+1$ nodes that include $J$ neighboring fog nodes and the connection to the cloud via the base station. For the \emph{cloud-centric bandwidth allocation}, the bandwidth allocated to the cloud is twice that of the bandwidth used by a fog node, i.e., the cloud and the fog node will be assigned the bandwidth $\frac{2B}{J+2}$ and $\frac{B}{J+2}$, respectively. 

Since the tasks arrive according to a Poisson process, and the transmission time in \eqref{eq:mu} is deterministic, the latency of the transmission queue can be modeled as an M/D/1 system\footnote{Instead of M/D/1 queueing, other delay  models can be used to account for other characteristics, such as different packet size or finite buffer size.} \cite{bertsekas1992data}:
\begin{eqnarray}\label{Tj}
T_j(\lambda_{ij}(\alpha_{ij}), \mu_{ij}) = \frac{\lambda_{ij}(\alpha_{ij})}{2\mu_{ij}(\mu_{ij}-\lambda_{ij}(\alpha_{ij}))} + \frac{1}{\mu_{ij}},
\end{eqnarray}
where the first term is the waiting time in the queue at fog node $i$, and the second term is the transmission delay between fog nodes $i$ and $j$. Similarly, when the tasks are offloaded to the cloud, the transmission queue delay will be: 
\begin{eqnarray}\label{Tc}
T_c(\lambda_c(\alpha_c), \mu_c) =  \frac{\lambda_{c}(\alpha_c)}{2\mu_{c}(\mu_{c}-\lambda_{c}(\alpha_c))} + \frac{1}{\mu_c},
\end{eqnarray}
where the service rate $\mu_c$ between fog node $i$ and cloud $c$ is given by \eqref{eq:mu} where fog node $j$ is replaced with cloud~$c$. 

Next, we define the computation queue. When a fog node needs to compute a task, this task will experience a waiting time in the computation queue of this fog node due to a previous task that is currently being processed. Since a fog node $j$ receives tasks from not only fog node $i$ but also other fog nodes and sensors, the task arrival process can be approximated by a Poisson process by applying the Kleinrock approximation \cite{bertsekas1992data}. Therefore, the computation queue can be modeled as an M/D/1 queue and the  latency of fog node $j$'s computation will be: 
\begin{eqnarray}\label{Sj}
	S_j(\lambda_{ij}(\alpha_{ij})) =  \frac{\lambda_{ij}(\alpha_{ij})}{2\mu_{j}(\mu_{j}-\lambda_{ij}(\alpha_{ij}))} + \frac{1}{\mu_{j}} + \omega_j \lambda_{ij}(\alpha_{ij}), 
\end{eqnarray} 
where the first term is the waiting delay in the computation queue, the second term is the delay for fetching the proper application  needed to compute the task, and the third term is a function of the processor delay implying the processing delay for the task. The delay of this fetching procedure depends on the performance of the node's hardware which is a deterministic constant that determines the service time of the computation queue. In the first and second terms of \eqref{Sj}, $\mu_j$ is a parameter related to the overall hardware performance of fog node $j$. In the third term, $\omega_j \lambda_{ij}(\alpha_{ij})$ is the actual computation time of the task with $\omega_j$ being a constant time needed to compute a task. For example, $1/\omega_j$ can be  proportional to the CPU clock frequency of fog node $j$. $\omega_j \lambda_{ij}(\alpha_{ij})$ implies that the delay needed to compute a task at a given node can increase with the task arrival rate since the number of concurrently running tasks increases with the task arrival rate. The increased number of the concurrently running tasks also increases the context switching delay that affects the computing delay. For fog node $j\in\cJ$, it is assumed that the maximum of computing service rate and processing speed are given by $\bar\mu_j$ and $1/\underline\omega_j$, respectively. This information can be known in advance if the manufacturers of fog devices can provide the hardware performance in the database. Then, when fog node $i$ locally computes its assigned tasks $\lambda_i(\alpha_i)$, the latency will be: 
\begin{eqnarray}\label{Si}
S_i(\lambda_i(\alpha_i)) =  \frac{\lambda_{i}(\alpha_i)}{2\mu_{i}(\mu_{i}-\lambda_{i}(\alpha_i))} + \frac{1}{\mu_{i}} + \omega_i  \lambda_i (\alpha_i),
\end{eqnarray}
where $\mu_i$ is the computing service rate of fog node $i$ (dependent on hardware performance) and $\omega_i \lambda_i(\alpha_{i})$ is the fog node $i$'s computing time. Since the cloud is equipped with more powerful and faster hardware than the fog node, the waiting time at the computation queue of the cloud can be ignored. This implies that the cloud initiates the computation for the received tasks without queueing delay; thus, we only account for the actual computing delay. As a result, when tasks are computed at the cloud, the computing delay at the cloud will be: 
\begin{eqnarray}\label{Sc}
S_c(\lambda_c(\alpha_c)) = \omega_c \lambda_c(\alpha_c).
\end{eqnarray}

In essence, if a task is routed to the cloud $c$, the latency will be 
\begin{eqnarray}\label{Dc}
 D_c(\lambda_c(\alpha_c), \mu_c) =  T_c(\lambda_c(\alpha_c), \mu_c) + S_c(\lambda_c(\alpha_c)).
\end{eqnarray}
Also, if a task is offloaded to fog node $j$, then the latency can be defined as the sum of the transmission and computation queueing delays: 
\begin{eqnarray}\label{Dj}
D_j(\lambda_{ij}(\alpha_{ij}), \mu_{ij}) = T_j(\lambda_{ij}(\alpha_{ij}), \mu_{ij}) + S_j(\lambda_{ij}(\alpha_{ij})).
\end{eqnarray}
Furthermore, when fog node $i$ computes tasks locally, the latency will be:\vspace{-4mm}
\begin{eqnarray}
D_i(\lambda_i(\alpha_i)) = S_i(\lambda_i(\alpha_i)),
\end{eqnarray}
since no transmission queue is necessary for local computing. Since $x_i$ is constant, $\lambda_{k\in\{i,{ij},c\}}$ is only dependent to $\alpha_k$. From now on, for notational simplicity, $\lambda_k(\alpha_k)$ is presented by $\lambda_k$. Given this model, in the next section, we formulate an online latency minimization problem to study how a fog network  can be formed and how tasks are effectively distributed in the fog network. 

\vspace{-3mm}
\section{Problem Formulation}\label{sec:problemformulation}

In distributed fog computing, the maximum latency of computing nodes must be minimized for effective distributed computing. To minimize the maximum latency, fog node $i$ must opportunistically find neighboring nodes to form a fog network and carry out the process of task offload. In practice, such neighbors will dynamically join and leave the system. Also, the neighbors have to process their existing workloads \cite{varghese2016challenges}. As a result, the initial fog node $i$ will be unable to know a priori whether an adjacent fog node will be available to assist it with its computation by sharing the communication and computational resources. Moreover, the total number of neighboring fog nodes as well as their locations and their available computing resources are unknown and highly unpredictable. Under such uncertainty, jointly optimizing the fog network formation and task distribution processes is challenging since selecting neighboring fog nodes must account for potential arrival of new fog nodes that can potentially provide a higher data rate and stronger computational capabilities. To cope with the uncertainty of the neighboring fog node arrivals while considering the data rate and computing capability of current and future fog nodes, we introduce an \emph{online optimization scheme} that can handle the problem of fog network formation and task distribution under uncertainty.

We formulate the following online fog network formation and task distribution problem whose goal is to minimize the maximum latency when computing a new task that arrives at fog node~$i$:
\begin{eqnarray}
\label{problem1}
\hspace{-3mm}\min_{ \cJ_{\boldsymbol\sigma}, \boldsymbol{\alpha}}   &&\hspace{-4mm}
                        \max \left(D_i(\lambda_i), \; D_c(\lambda_c, \mu_c), \; D_{j \in \cJ_{\boldsymbol\sigma}}(\lambda_{ij}, \mu_{ij})\right),\\
\hspace{-3mm}\textrm{s.t.} &&\hfill \nonumber\\
\hspace{-3mm}									&&\hspace{-10mm} \alpha_i +\alpha_c  +\textstyle \sum_{j\in\cJ}  \alpha_{ij} = 1, \label{problem1:c1}\\
\hspace{-3mm}                   &&\hspace{-10mm} \alpha_i  \in [0,1], \alpha_c \in [0,1], \alpha_{ij}  \in [0,1], \forall j \in \cJ_{\boldsymbol\sigma} \subset \cN_{\boldsymbol\sigma}, \label{problem1:c2}\\
\hspace{-3mm}                   &&\hspace{-10mm} \alpha_i x_i \!\leq\! \mu_i, \! \alpha_c x_i \!\leq\! \mu_c, \! \alpha_{ij} x_i \!\leq\! \mu_{j}, \! \alpha_{ij} x_i \!\leq\! \mu_{ij},  \! \forall j\!\in\!\cJ_{\boldsymbol\sigma},\label{problem1:c4}\\
\hspace{-3mm}                   &&\hspace{-10mm}  |\cN_{\sigma}| \leq N. \label{problem1:c3} 
\end{eqnarray}
Since our goal is to minimize the worst-case latency among the fog nodes and the cloud, any task can be processed with a low latency regardless of which node actually computes the task\footnote{If the objective function is defined with a minimum function, the initial fog node will minimize the latency of only one node, and, therefore, it will increase the latency of other nodes.}. By using an auxiliary variable $u$, problem \eqref{problem1} can be transformed into the following: 
\begin{eqnarray}\label{problem2}
\hspace{-3mm}\min_{\color{black}  \cJ_{\boldsymbol\sigma}, \boldsymbol{\alpha}}   && \hspace{-5mm}  u,\\
\hspace{-3mm}\textrm{s.t.} && \hspace{-5mm} u \geq \max \left(D_i(\lambda_i),  D_c(\lambda_c, \mu_c),  D_{j \in \cJ_{\boldsymbol\sigma}}(\lambda_{ij}, \mu_{ij})\right), \label{problem2:c1}\\
\hspace{-3mm}                   && \hspace{-5mm} \eqref{problem1:c1}, \eqref{problem1:c2}, \eqref{problem1:c4}, \eqref{problem1:c3},\nonumber 
\end{eqnarray}
where $u$ is the maximum latency of the fog network. In \eqref{problem2}, $u$ represents the largest value among $D_i(\lambda_i), D_c(\lambda_c, \mu_c)$, and $D_{j}(\lambda_{ij}, \mu_{ij})$. Then, minimizing $u$ is equivalent to minimizing the $\max$ function in \eqref{problem1}. Hence, problems \eqref{problem1} and \eqref{problem2} are equivalent. 

In constraints~\eqref{problem1:c1}~and~\eqref{problem1:c2}, all tasks arriving at fog node $i$ are offloaded among the computing nodes in the fog network. Due to constraint~\eqref{problem1:c4}, the tasks offloaded to a node cannot exceed the service rate of the computing node. In this problem, the initial fog node $i$ determines the set of neighboring fog nodes $\cJ_{\boldsymbol\sigma}$ when they arrive online and the task distribution vector $\boldsymbol\alpha$ so as to minimize the computing latency. Fog node $i$ will observe a total number of $N$ arriving fog nodes due to constraint \eqref{problem1:c3}. Fog node $i$ has to make a decision on network formation and task distribution while observing $N$ neighboring nodes. As the number of observations increases, fog node $i$ may be able to discover neighboring fog nodes that have higher performance. However, due to constraint \eqref{problem1:c3}, fog node $i$ cannot wait to observe an infinite number of neighboring fog nodes. Thus, while observing up to $N$ arriving fog nodes, fog node $i$ should select $J\leq N$ neighboring fog nodes to minimize \eqref{problem1}. 

In our model, we assume that fog node~$i$ does not have any prior information on the neighboring fog nodes given by set $\cN_{\boldsymbol\sigma}$, and the information about each neighboring node is collected sequentially. Such random arrival sequence is denoted by $\boldsymbol\sigma = \sigma_1, \ldots, \sigma_n, \ldots, \sigma_N$ where the arrival of $n$-th neighboring node is shown as $\sigma_n$. For example, a smartphone can choose to become a fog node spontaneously if it decides to share its resources. In practice, to discover the neighboring nodes, the fog nodes can use the node discovery mechanisms implemented in D2D networks \cite{vallati2015exploiting}. When fog node~$i$ does not have complete information on other fog nodes, the nodes in  $\cN_{\boldsymbol\sigma}$ arrive at fog node $i$ in a random order, and index $n$ can be the arriving order of the neighboring fog nodes. At the arrival of a neighboring node, the arrival order $n$ increases by one; thus, $n$ captures the time order of arrival. At time $n$, node $n$ can transmit a beacon signal to fog node $i$ to indicate its willingness to join the network of fog node $i$. The beacon signal can include an information tuple on node $n$ that includes the distance $d_{in}$, computing service rate $\mu_n$, and the processing speed $\omega_n$. At each time that $\sigma_n$ is known, e.g., by receiving the beacon signal, fog node $i$ will now have information on these parameters that pertain to node $n$ \cite{doppler2011advances}. Therefore, fog node $i$ only knows the information on the nodes that have previously arrived (as well as the current node). 

When fog node $i$ observes $\sigma_n$ and has knowledge of the  $n$-th neighboring node, it has to make an online decision whether to select node $n$. If fog node $n$ is chosen by the initial fog node $i$, it is indexed by $j$ and included in a set $\cJ_{\boldsymbol\sigma}$ which is a subset of $\cN_{\boldsymbol\sigma}$. Otherwise, fog node $i$ will no longer be able to select fog node $n$ at a later time period since the latter can join another fog network or terminate its resource sharing offer to fog node~$i$. For notational simplicity, $\cJ_{\boldsymbol\sigma}$ and $\cN_{\boldsymbol\sigma}$ are hereafter denoted as $\cJ$ and $\cN$, respectively. Fog node $i$ will not be able to have complete information about all $N$ neighboring nodes before all neighboring nodes are selected by fog node $i$. Therefore, since fog node $i$ cannot know any information on future fog nodes, it is challenging for the initial fog node $i$ to form the fog network by determining~$\cJ$. 

Even when the information on each node is known to fog node $i$, it is difficult to calculate the exact service rates of the fog node in the formulated problem. This is due to the fact that the service rate in \eqref{eq:mu}, that includes the wireless data rate, is a function of the network size $J$. As the number of nodes sharing their wireless bandwidth increases, the available channel bandwidth per node decreases, thus reducing the data rate. Therefore, unlike the constant parameters  $\mu_i$ and $\mu_j$, the transmission service rates $\mu_{ij}$ and $\mu_c$ will vary with the network size. As a consequence, in order to calculate the service rates of neighboring nodes, fog node $i$ has to determine the network size. However, the optimal network size can change by the selection of neighboring nodes. Since network size and node selection are related, it is challenging for fog node $i$ to optimize both network size and the set of neighboring nodes that minimize \eqref{problem2}. To solve the online problem, we need to find the set of neighboring fog nodes $\cJ$ and the task distribution vector $\boldsymbol{\alpha}$ that minimize the maximum latency. Moreover, since there is uncertainty about the future arrival of neighboring nodes as well as their service rates, one has to seek an online, sub-optimal solution that is also robust to uncertainty. In the next section, we propose an online optimization framework that  minimizes the value of $u$ in \eqref{problem2}. 

\vspace{-3mm}
\section{Task Distribution and Network Formation Algorithms}\label{sec:solution}

In our problem, fog node $i$ has to decide whether to admit each neighboring node as the different neighboring nodes arrive in a random order. This problem can be formulated as an online stopping problem. In such problems, such as online secretary problem \cite{babaioff2007knapsack}, the goal is to develop online algorithms that enable a company to hire a number of employees, without knowing in which order they will arrive to the interview. To apply such known solutions from the stopping problems, the following assumptions are commonly needed. For instance, the number of hiring positions should be deterministic and given in the problem. Also, the decision maker should be able to decide the preference order among the candidates by comparing the values that can be earned by hiring candidates. Under these assumptions, online stopping algorithms can be used to select the best set of candidates in an online manner. In this regard, even though the structures of our fog network formation problem and the secretary problem are similar, the fog network formation case has different assumptions. First, the number of neighboring fog nodes is an optimization variable in our problem. Second, the latency of computing nodes that somewhat maps to the valuation of hiring candidates in the secretary problem is not constant. Moreover, in our problem, each neighboring fog node exhibits two types of latency: transmission latency and computing latency. As a result, it is challenging to define the preference order of the neighboring nodes as done in conventional online stopping problems. To address those challenges, we propose a new online optimization framework\footnote{The framework proposed in this work is different from the previous work in \cite{lee2017online} since this work uses a different definition of transmission service rate in \eqref{eq:mu} and a different objective function in \eqref{problem1}.} that extends existing results from online stopping theory to accommodate the specific challenges of the fog network formation problem\footnote{Fog networks can be formed by using game-theoretic approaches such as coalitional games which require a complete knowledge of the exact utility functions \cite{saad2009distributed}. However, such  knowledge can be difficult to gather, since the initial fog node cannot have the complete information on the neighboring nodes in an online scenario, and, therefore, an online optimization framework is more apropos. Moreover, using a coalitional game framework to solve the proposed fog network formation problem under uncertainty will require the use of very complex algorithms that are not amenable to analysis, unlike the proposed online optimization framework.}.

\vspace{-3mm}
\subsection{Overview of the Proposed Optimization Framework}

\begin{figure}[]
\centering
\includegraphics[width=0.35\textwidth]{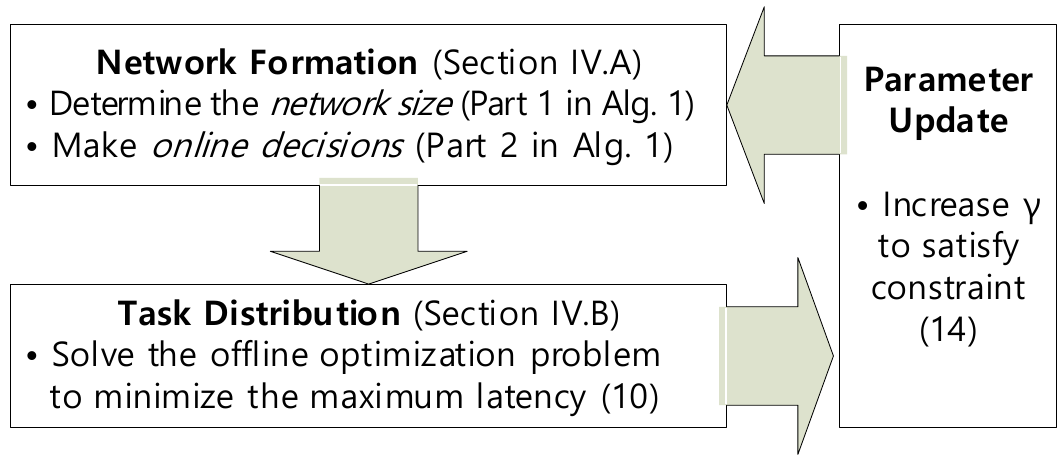}
\caption{Online optimization framework for Fog network formation and task distribution.}
\label{fig:framework}\vspace{-3mm}
\end{figure}

Problem~\eqref{problem2} has two optimization variables $\cJ$ and $\boldsymbol{\alpha}$ that constitute the solutions of the network formation and task distribution problems, respectively. To solve \eqref{problem2}, fog node $i$ must first optimize the network formation by selecting the neighboring fog nodes, and then decide on its task distribution. This two-step process is required due to the fact that the computing resources of the fog nodes are unknown before the network is formed. The  online optimization framework consists of three highly inter-related components as shown in Fig.~\ref{fig:framework}. In the \emph{network formation stage}, an online algorithm is used to find $\cJ$ by determining the minimal network size and, then, selecting the neighboring fog nodes within $N$ observations to satisfy \eqref{problem1:c3}. After $\cJ$ is determined, the task distribution among the selected nodes is optimized by using an offline optimization method during the \emph{task distribution stage}. The output of the task distribution stage is the task allocation vector $\boldsymbol\alpha$ that satisfies constraints \eqref{problem1:c1}, \eqref{problem1:c2}, and \eqref{problem1:c4}. Finally, we use a  \emph{parameter update stage}, during which the target performance parameter $\gamma$ that will be used in the next iteration is updated in order to satisfy constraint~\eqref{problem1:c3}. After repeatedly running three components of our framework, fog node $i$ is able to form a network without any prior information on the neighboring nodes and also offload the tasks to the nodes on the fog network. This algorithm is shown to converge in Theorem~\ref{theorem3}. 

The performance of our online optimization framework will be evaluated by using \emph{competitive analysis} \cite{borodin2005online}. In this analysis, the performance is measured by the \emph{competitive ratio} $\gamma$ that is defined by 
\begin{eqnarray}\label{cr}
1\leq \frac{\textrm{ALG}(\boldsymbol\sigma)}{\textrm{OPT}(\boldsymbol\sigma)} \leq \gamma,
\end{eqnarray}
where $\textrm{ALG}(\boldsymbol\sigma)$ denotes the latency achieved by the online algorithm and $\textrm{OPT}(\boldsymbol\sigma)$ is the optimal latency achieved by an offline algorithm. If the online algorithm finds the optimal solution, the online algorithm achieves $\gamma=1$. However, since the online algorithm cannot have complete information, it is challenging to find the optimal solution in an online setting. Therefore, in an online minimization problem, the online algorithm should be able to achieve $\gamma$ that is close to one. We use this notion of competitive ratio to design our online optimization~framework.

\begin{figure}[]
\centering
\includegraphics[width=0.30\textwidth]{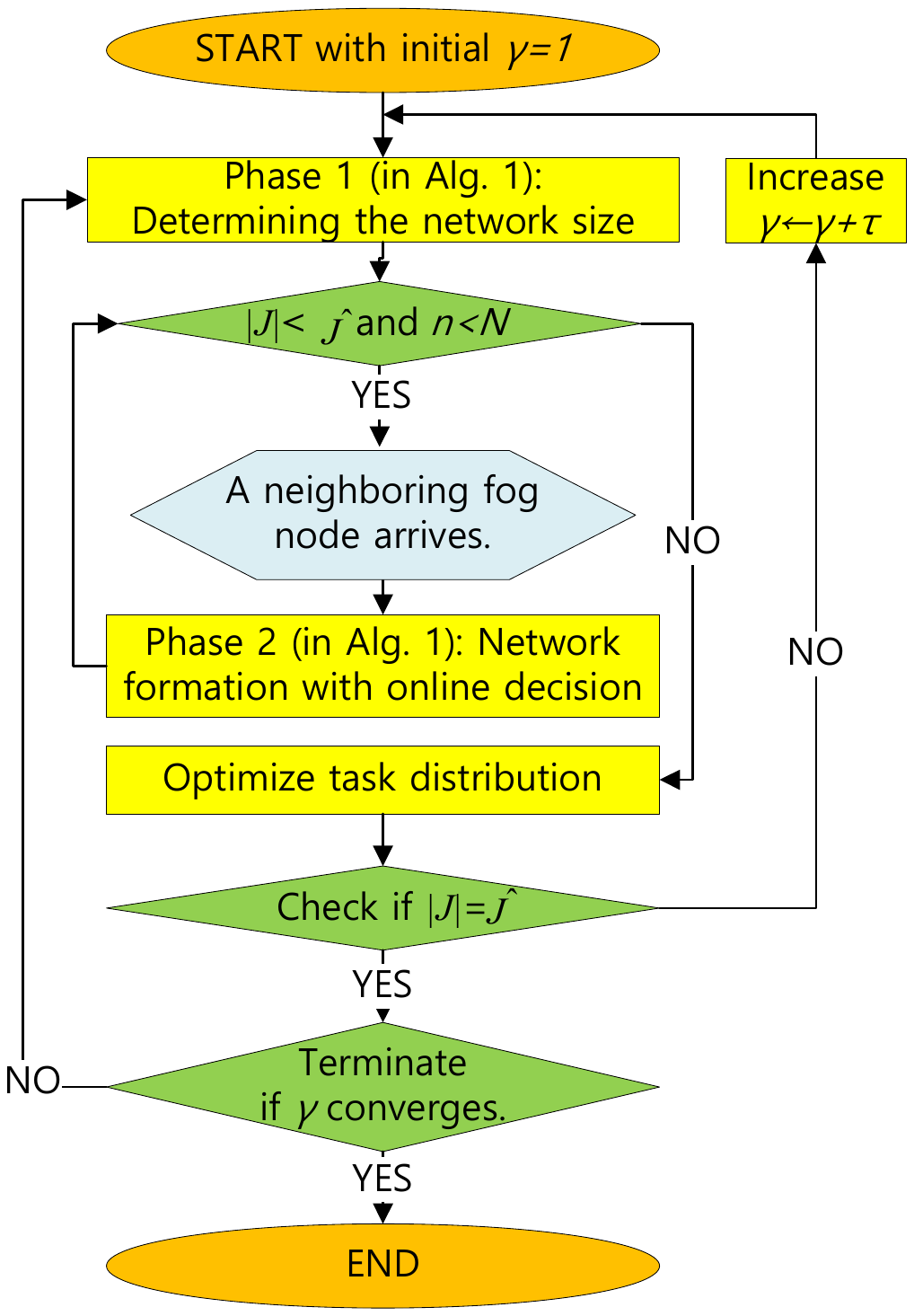}
\caption{ Flow chart of the proposed  framework for fog network formation and task distribution.}
\label{fig:flowchart}\vspace{-3mm}
\end{figure}

The  online optimization framework is summarized in the flow chart shown in Fig.~\ref{fig:flowchart}. In the network formation stage, fog node $i$ needs to select the set of neighboring fog nodes with high service rates and processing speeds  to achieve a given value of $\gamma$. At each iteration, to achieve a target competitive ratio $\gamma$, fog node $i$ determines the number of neighboring nodes $\hat{J}$ by using Phase 1 of Algorithm~\ref{algorithm}, and it sequentially observes the arrivals of a total of $N$ neighboring fog nodes while making an online decision in Phase 2 of Algorithm~\ref{algorithm}. After the network formation stage is finished, the task distribution is optimized by the initial fog node in an offline manner. Then, fog node $i$ checks whether the number of selected neighboring nodes is $\hat{J}$. For a small value of $\gamma$, fog node $i$ must find the neighboring nodes having a high computing service rate and processing speed so as to achieve low latency. Therefore, in this case, fog node~$i$  must observe a large number of neighboring nodes until $\hat{J}$ neighboring nodes are selected. Hence, $N$ observations may not be sufficient to find $\hat{J}$ neighboring nodes. On the other hand, a large $\gamma$ can allow the target latency to be less stringent, thus allowing the fog node $i$ to  select the neighboring nodes with fewer observations. To find the proper value of $\gamma$, the proposed  framework iteratively updates $\gamma$. For instance, the value of $\gamma$ can be set to one initially. Then, if a smaller~$\gamma$ cannot be achieved in the network formation stage at that iteration, the value of~$\gamma$ increases by a small constant~$\tau$. By repeatedly increasing $\gamma$, the proposed  framework can find the achievable value of $\gamma$. In the next section, we present the details of the proposed online algorithm that exploits the updated value of $\gamma$ for the network formation stage. 

\vspace{-3mm}
\subsection{Fog Network Formation: Online Approach}\label{sec:algorithm}

\begin{algorithm}[t]\smallskip\smallskip
\caption{Online Fog Network Formation Algorithm} 
\begin{algorithmic}[1]\label{algorithm}\vspace{-2mm}\small	
\item[1 :] \hspace{0.0cm}  {\bf inputs}: {\color{black}$N$}, $\gamma$, $\mu_i$, $\omega_i$, $\omega_c$, $d_c$, $\bar\mu_{ij}(\underline d_{ij})$, $\bar\mu_j$, $\underline\omega_j$. 
\item[   ] \hspace{0.0cm}  \emph{Phase 1: Calculate $\hat{\lambda}_{ij}$, $\hat{J}$, and $\hat{u}$.}
\item[2 :] \hspace{0.3cm}  {\bf initialize}: $J=0$, {\color{black}$n=0$}.
\item[3 :] \hspace{0.3cm}  {\bf while} $\Delta \geq 0$
\item[4 :] \hspace{0.6cm}    $J\leftarrow J+1$.
\item[5 :] \hspace{0.6cm}  $\Delta \leftarrow \left[D_j(\lambda_{ij}, \bar\mu_{ij})\right]_{|\cJ|=J-1} - \left[D_j(\lambda_{ij}, \bar\mu_{ij})\right]_{|\cJ|=J}$. 
\item[6 :] \hspace{0.3cm}  {\bf end while}
\item[7 :] \hspace{0.3cm}  Find $\hat{\lambda}_{ij}$ by optimizing task distribution when $|\cJ|=J-1$.
\item[8 :] \hspace{0.3cm}  Set $\hat{J}=J-1$ and $\hat{u}= \left[ D_j(\hat{\lambda}_{ij}, \bar\mu_{ij}) \right]_{|\cJ|=J-1}$.
\item[   ] \hspace{0.0cm}  {\emph{Phase 2: Decide $\cJ$.}} 
\item[9 :] \hspace{0.3cm}  {\bf while} {\color{black} $|\cJ| < \hat{J}$ and $n < N$ }
\item[10:] \hspace{0.6cm}    {\bf if} $D_n(\hat{\lambda}_{ij}, \mu_{in}) \leq \gamma  \hat{u}$,
\item[11:]\hspace{0.9cm}       $\cJ  \leftarrow  \cJ \cup \{n\}$.
\item[12:]\hspace{0.6cm}     {\bf end if}
\item[13:]\hspace{0.6cm}     {\color{black}$n \leftarrow n+1$.}
\item[14:]\hspace{0.3cm}   {\bf end while}\vspace{-2mm}
\end{algorithmic}\smallskip\smallskip
\end{algorithm}

In problem~\eqref{problem2}, the decision on $\cJ$ faces two primary challenges: how many fog nodes are needed in the network and which fog nodes join the network (at which time). Since the transmission service rates are functions of the wireless bandwidth that can vary with the network size, the service rates of neighboring fog nodes cannot be calculated without having a fixed network size. Therefore, the proposed  algorithm includes two phases as shown in Algorithm~\ref{algorithm}. The goal of the first phase is to determine the parameters including the network size and the temporal task distribution so that the parameters can be used in the second phase of Algorithm~\ref{algorithm}. Then, the second phase of Algorithm~\ref{algorithm} allows fog node $i$ to make an online decision regarding the selection of an  arriving node. 

In the first phase of Algorithm~\ref{algorithm}, the goal is to determine the parameters that will be used in the second phase of Algorithm~\ref{algorithm}. In the given system model, a neighboring node will be referred to as \emph{ideal} in terms of minimizing the latency in \eqref{problem2} if it has the highest computing service rate $\bar\mu_j$, processing speed $1/\underline\omega_j$, and transmission service rate $\bar\mu_{ij}$ when the distance between two fog nodes is $\underline d_{ij}$. Such an \emph{ideal node} is denoted by $\bar{j}$. If a network is formed with  nodes having high computing resources, a smaller network size can effectively minimize the latency. When the service rates of the nodes are divided by the smallest network size, the transmission service rates of the nodes also can be maximized, and, hence, the latency can be minimized. In the case in which the ideal nodes construct a network, the minimized latency of \eqref{problem2} is denoted by $\hat{u}$. Also, when the latency is $\hat{u}$, the corresponding number of neighboring nodes and task distribution are denoted by $\hat{J}$ and $\{\hat{\lambda}_{i}, \hat{\lambda}_{c}, \hat{\lambda}_{ij}\}$, respectively. 

{\bf{First phase:}} The first phase of Algorithm~\ref{algorithm} is used to calculate $\hat{J}$ and $\hat{\lambda}_{ij}$. The latency in \eqref{problem2} decreases as the number of neighboring nodes increases since the computational load per node can be reduced. However, if the number of neighboring nodes becomes too large, the bandwidth per fog node will be smaller yielding  lower  transmission service rates for the nodes. Consequently, the latency can increase with the number of neighboring nodes, due to these bandwidth limitations. By using the relationship between network size and latency, the first phase of Algorithm~\ref{algorithm} searches for $\hat{J}$ while increasing the network size incrementally, one by one. Once the number of neighboring users~$\hat{J}$ that minimizes $\hat{u}$ is found, the tasks offloaded to each ideal node are denoted by  $\hat\lambda_{ij}$. Therefore, we will have $\hat{J}$, $\hat{u}$, and $\hat\lambda_{ij}$ as the outputs from the first phase of Algorithm~\ref{algorithm} that will be used in the second phase of Algorithm~\ref{algorithm}. 

{\bf{Second phase:}} In the second phase of Algorithm~\ref{algorithm}, fog node $i$ decides on whether to select each neighboring node or not, by using a threshold-based algorithm. Our algorithm uses a single threshold so that the latency of each arriving node can be compared with the threshold value. Since comparing two values is a simple operation having constant time complexity, a threshold-based algorithm can be executed with low latency. However, before the network formation process is completed, fog node $i$ is not able to know the optimal latency of each node, and, therefore, finding the distribution of tasks that must be offloaded to each node is not possible. Nonetheless, fog node $i$ must set a threshold before the first neighbor arrives. To this end, fog node~$i$ sets this initial  threshold by assuming that an equal amount of tasks, $\hat{\lambda}_{ij}$, is offloaded to each one of the $\hat{J}$ neighboring nodes. Thus, in our threshold-based algorithm, the threshold value is compared with the latency that results from offloading $\hat{\lambda}_{ij}$ tasks. For example, when a neighboring node $n$ arrives, the algorithm compares the latency of node $n$, $D_n(\hat{\lambda}_{ij}, \mu_{in})$, to the threshold $\gamma \hat{u}$. If the latency of node $n$ is smaller than the threshold, fog node $i$ will immediately select node $n$. This procedure is repeated until fog node $i$ observes \!$N$ arrivals and selects \!$\hat{J}$ neighboring nodes. In the  proposed algorithm, the initial fog node needs to discover the neighboring nodes and know the information on the communication and computational performance of the neighboring nodes. This procedure can use any node-discovery and message exchanging protocols developed for D2D communications or wireless sensor networks. Also, our framework requires a low signaling and communication overhead  since each neighboring node can transmit its location and computing speed using a very small packet after which the initial fog node transmits a decision on node selection using a single bit. After the fog network is formed, the task distribution is done to minimize latency. In the next section, we investigate the property of the optimal task distribution, and show that the threshold can satisfy \eqref{cr}. 
 
 \vspace{-3mm}
\subsection{Task Distribution: Offline Optimization}\label{sec:offline}

Once the nodes are selected to form a network, the task distribution can be performed using an offline optimization problem which can be solved using known algorithms such as the interior-point algorithm \cite{nocedal2006numerical}. From problem \eqref{problem2}, the following properties can be derived, for a given~$\cJ$. 
\begin{theorem}\label{theorem1}
If there exists a task distribution $\boldsymbol{\alpha}^*$ satisfying $u^*=D_i(\lambda_i) = D_c(\lambda_c, \mu_c) = D_j(\lambda_{ij}, \mu_{ij})$, $\forall j \in \cJ$, then $\boldsymbol{\alpha}^*$ is the unique and optimal solution of problem~\eqref{problem1}.
\end{theorem}
\begin{proof}
Let $\boldsymbol{\alpha}$ be the initial task distribution, and assume that any other task distribution  $\boldsymbol{\alpha'}$ different from  $\boldsymbol{\alpha}$ is the optimal distribution. When $\boldsymbol{\alpha'}$ is considered, we can find a certain node $A$ satisfying $\alpha'_A < \alpha_A$ where $\alpha'_A \in \boldsymbol{\alpha'}$ and $\alpha_A \in \boldsymbol{\alpha}$. This, in turn, yields $D_A(\alpha'_A) < D_A(\alpha_A)$. Due to the constraint~\eqref{problem1:c1}, there exists another node $B$ such that $B\neq A$, $\alpha'_B > \alpha_B$, and $D_B(\alpha'_B) > D_B(\alpha_B)$ where $\alpha'_B \in \boldsymbol{\alpha'}$ and $\alpha_B \in \boldsymbol{\alpha}$. Since $D_B(\alpha'_B) > D_B(\alpha_B) =  D_A(\alpha_A) > D_A(\alpha'_A)$, we must decrease $\alpha'_B$ to minimize the maximum, i.e., $D_B(\alpha'_B)$. Thus, we can clearly see that  $\boldsymbol{\alpha'}$ is not optimal, and, thus, the initial distribution $\boldsymbol{\alpha}$ is optimal.

Furthermore, $D_j(\lambda_{ij},\mu_{ij})$ is a monotonically increasing function with respect to $\lambda_{ij}=x_i \alpha_{ij}$ since $\frac{\partial}{\partial \lambda_{ij}}D_j(\lambda_{ij},\mu_{ij}) > 0$. Therefore, there are no more than two points of $\boldsymbol{\alpha}^*$ that have the same $u^*$. Hence, the distribution $\boldsymbol{\alpha}$ is unique and optimal. 
\end{proof}
Theorem~\ref{theorem1} shows that the optimal solution of the offline latency minimization problem results in an equal latency for all fog nodes and the cloud on the network (whenever such a solution  is feasible). Using the objective function in \eqref{problem1}, the initial fog node minimizes the worst-case latency among the nodes. To that end, the initial fog node can decrease the task arrival rate of the node having the highest latency, but, in turn, the latency of other node increases. This is due to the fact that reducing one node's task arrival rate leads to increase the other node's arrival rate since we have $\sum_{j\in \cJ} \lambda_{ij} + \lambda_i + \lambda_c = x$. Therefore, as shown in Theorem~\ref{theorem1}, an equal latency for all fog nodes and the cloud is obtained by repeatedly reducing the arrival rate of the node having the highest latency. According to Theorem~\ref{theorem1}, selecting the node that has high computing resources is beneficial to minimize latency. Once fog node $i$ determines the task distribution, the efficiency of the task distribution can be derived by applying the definition of task scheduling efficiency in \cite{mirshekarian2016correlation}. For a task distribution $\boldsymbol{\alpha}$, the \emph{efficiency} is given by 
\begin{equation}\label{efficiency}
 \Gamma = 1+ \frac{  \displaystyle{\sum_{k\in\{i,c, \{ij| j\in\cJ\}\}}} \hspace{-4mm} \max \Bigg(\splitfrac{D_i(\alpha_i),}{\hspace{-4mm}\splitfrac{D_c(\alpha_c, \mu_c),}{\hspace{-3mm}D_{j\in\cJ}(\alpha_{ij}, \mu_{ij})}} \Bigg) - D_k}{D_i(\alpha_i) + D_c(\alpha_c) + \sum_{j\in\cJ}D_j(\alpha_{ij})}  \geq 1. 
\end{equation}
In other words, $\Gamma$ is defined as one plus the ratio between the total idle time of the fog computing nodes and the total transmission and computing time. Therefore, $\Gamma=1$ means that all nodes in the fog network can complete their assigned tasks with the same latency. Theorem~\ref{theorem1} shows that the optimal latency is $u^*=D_i(\lambda_i) = D_c(\lambda_c, \mu_c) = D_j(\lambda_{ij},\mu_{ij})$. Since $u^*$ is the maximum value among $D_i(\lambda_i)$, $D_c(\lambda_c,\mu_c)$, and $D_j(\lambda_{ij}, \mu_{ij})$, from \eqref{problem1}, the efficiency of the optimal task distribution will be equal to one. Thus, if the efficiency of the task distribution becomes one, the latency of the task distribution is the optimal latency $u^*$ according to Theorem~\ref{theorem1}. 

\subsection{Performance Analysis of the Proposed Online Optimization Framework}

Next, we show that the proposed framework can achieve the target competitive ratio~$\gamma$. 
\begin{theorem}\label{theorem2}
For a given $\gamma$, the proposed framework satisfies ${\textrm{ALG}(\boldsymbol\sigma)}/{\textrm{OPT}(\boldsymbol\sigma)} \leq  \gamma$ if: (i) a given $\gamma$ enables fog node $i$ to select $\hat{J}$ nodes, and (ii) the optimal task distribution can always be found, i.e., $\Gamma=1$.  
\end{theorem}
\begin{proof}

The offline optimal latency of the nodes in $\cJ$ is greater than or equal to $\hat{u}$, i.e., $\hat{u} \leq \textrm{OPT}(\boldsymbol\sigma)$. Also, in Algorithm~\ref{algorithm}, the selected nodes satisfy $D_j(\hat{\lambda}_{ij}, \mu_{ij}) \leq \gamma  \hat{u}$, $\forall j\in\cJ$ where $|\cJ|=\hat{J}$. When the task distribution is not yet optimized with respect to $\cJ$, the latency that results from using distribution $\!\{\hat{\lambda}_{i}, \!\hat{\lambda}_{c}, \!\hat{\lambda}_{ij}\}\!$ can be shown as $\textrm{ALG}_b(\boldsymbol{\sigma}) \!=\! \max\! \left(\! D_i(\hat{\lambda}_{i}), \!D_c(\hat{\lambda}_{c}, \mu_c), \!D_{j \in \cJ}(\hat{\lambda}_{ij}, \mu_{ij})\!\right)\!$. Recall that $\hat{u}  \triangleq \max \left(D_i(\hat{\lambda}_{i}), \; D_c(\hat{\lambda}_{c}, \mu_c), \; D_{\bar{j}}(\hat{\lambda}_{ij}, \mu_{i\bar{j}})\right)$, and, by Theorem~\ref{theorem1}, $\hat{u}= D_i(\hat{\lambda}_{i}) = D_c(\hat{\lambda}_{c}, \mu_c) = D_{j}(\hat{\lambda}_{ij}, \bar\mu_{ij})$. Since the service rates and computing speeds of selected node $j\in\cJ$ are less than or equal to those of the ideal node, i.e, $\mu_{ij} \leq \bar\mu_{ij}$, $\mu_{j} \leq \bar\mu_{j}$, and $1/\omega_j \leq 1/\underline\omega_j$, we have $\hat{u} \leq  D_{j \in \cJ}(\hat{\lambda}_{ij}, \mu_{ij})$. Therefore, we have $\textrm{ALG}_b(\boldsymbol{\sigma})
= \max \left(\hat{u}, \; D_{j}(\hat{\lambda}_{ij}, \mu_{ij})\right) \nonumber
= \max \left(D_{j}(\hat{\lambda}_{ij}, \mu_{ij})\right) \nonumber
\leq \gamma \hat{u}, \forall j\in\cJ$. 
By optimizing the task distribution for the nodes in $\cJ$, the latency can be further reduced, i.e, $\textrm{ALG}(\boldsymbol\sigma) \leq \textrm{ALG}_b(\boldsymbol\sigma)$. Hence, it is possible to conclude that $\textrm{ALG}(\boldsymbol\sigma) \leq \textrm{ALG}_b(\boldsymbol{\sigma}) \leq \gamma  \hat{u} \leq \gamma  \textrm{OPT}(\boldsymbol\sigma)$ and, therefore, ${\textrm{ALG}(\boldsymbol\sigma)}/{\textrm{OPT}(\boldsymbol\sigma)} \leq \gamma$. 
\end{proof}\vspace{-2mm}
\noindent This result shows that the online optimization framework can achieve the target competitive ratio $\gamma$ by determining a proper number of neighboring nodes $\hat{J}$ and optimizing the task distribution. According to Theorem~\ref{theorem2}, the ratio between the latency achieved by executing one iteration of the proposed framework and an offline optimal latency can be bounded by the value of $\gamma$. 

To satisfy the first condition of Theorem~\ref{theorem2}, the proper value of $\gamma$ needs to be found iteratively as shown in Fig.~\ref{fig:flowchart}. Then, we prove that $\gamma$ converges to an upper bound. For this proof, we define the lowest transmission service rate as $\underline \mu_{ij}$ when the maximum of $d_{in}$ is $\bar{d}_{ij}$. Also, the lowest computing service rate and the lowest processing speed are defined as $\underline \mu_{j}$ and $1/\bar\omega_{j}$, respectively.
\begin{theorem}\label{theorem3}
The target competitive ratio $\gamma$ converges to ${ D_j(\hat{\lambda}_{ij}, \underline\mu_{ij})}/{\hat{u}}$ if: (i) a given $\gamma$ enables fog node $i$ to select $\hat{J}$ nodes, and (ii) the optimal task distribution can always be found, i.e., $\Gamma\!=\!1$.
\end{theorem}
\begin{proof}
We show that there exists an upper bound of $\gamma$ denoted by $\bar\gamma$. Therefore, for a given sequence $\boldsymbol\sigma$, we show that $
 \frac{\textrm{ALG}(\boldsymbol\sigma)}{\textrm{OPT}(\boldsymbol\sigma)} \leq 
  \frac{\max_{\boldsymbol\sigma'}\textrm{ALG}(\boldsymbol\sigma')}{\min_{\boldsymbol\sigma'}\textrm{OPT}(\boldsymbol\sigma')}=\bar\gamma,
$ where $\boldsymbol\sigma'$ denotes any sequence. In the first phase of Algorithm~\ref{algorithm}, since $\hat{u}$ is calculated by assuming that all neighboring nodes are ideal nodes, the lower bound of the offline latency for any sequence is given by $\min_{\boldsymbol\sigma'}\textrm{OPT}(\boldsymbol\sigma')=\hat{u}$. Also, if $\hat{J}$ neighboring nodes are located at the farthest distance $\bar d_{ij}$, the lowest fog transmission service rate denoted as $\underline\mu_{ij}$ is derived. Then, the worst case is defined by assuming that the neighboring nodes have the lowest service rates and computing speed, i.e., $\underline\mu_{ij}$, $\underline\mu_{j}$, and $1/\bar\omega_{j}$. Therefore,  the latency in the worst case can be presented by $\max_{\boldsymbol\sigma'}\textrm{ALG}(\boldsymbol\sigma') = D_j(\hat{\lambda}_{ij}, \underline\mu_{ij})$. Finally, $\gamma$ always increases when it is updated, and, hence, $\gamma$ converges to a competitive ratio given by $\bar{\gamma}=\frac{ D_j(\hat{\lambda}_{ij}, \underline\mu_{ij})}{\hat{u}}$. 
\end{proof}
\noindent Therefore, the proposed  framework is able to find the target competitive ratio by iteratively updating $\gamma$ when  $\bar{d}_{ij}$, $\underline \mu_{j}$, and $1/\bar\omega_{j}$ are not known to fog node $i$. Thus, once $\gamma$ is found through the iterative process, Algorithm~\ref{algorithm} is used to select the neighboring nodes, and the tasks  are offloaded to the neighboring nodes as stated in Theorem~\ref{theorem1}. As a result, the proposed framework yields the set of $\hat{J}$ selected neighboring nodes and the corresponding task distribution that can achieve the target competitive ratio as shown in Theorem~\ref{theorem2}. 

The upper bound in Theorem~\ref{theorem3} is the performance in the worst case if a given $\gamma$ enables fog node $i$ to select $\hat{J}$ neighboring nodes, and the optimal task distribution can always be found, i.e., $\Gamma=1$. If  the first condition on the network size in Theorem~\ref{theorem3} cannot be satisfied, $\gamma$ is updated. When the target competitive ratio $\gamma$ converges to $\bar{\gamma}$, the number of iterations tends to infinity since the value of $\gamma$ asymptotically approaches to $\bar{\gamma}$. In particular, as $\gamma$ becomes closer to $\bar{\gamma}$, the probability of updating $\gamma$ decreases exponentially. Therefore, after running a finite, large number of iterations, the probability of updating $\gamma$ can become marginal. When the current value of $\gamma$ is rarely updated, the first condition on the network size in Theorem~\ref{theorem3} is assumed to be satisfied, and, thus, the iteration process used to update $\gamma$ will terminate. In doing so, the final value of $\gamma$ that is smaller than $\bar{\gamma}$ can be used to further reduce the latency of the formed fog network.  

To this end, we derive a lower bound of the probability, with respect to $\gamma$, that the initial fog node forms a fog network with $\hat{J}$ neighboring nodes in an iteration including $N$ observations. To derive a statistical result, we assume that the values of the communications and computing capabilities of neighboring nodes are random variables. For example, the distance, $d_{in}$, between the initial node and a neighboring node is a random variable within a finite range $[\underline d_{ij}, \bar d_{ij}]$, and, therefore, the service rate $\mu_{in}$ from \eqref{eq:mu} is a random variable in the range $[\underline \mu_{ij}, \bar \mu_{ij}]$. Also, a neighboring node's computing service rate $\mu_n$ and computing delay $\omega_n$ can be modeled as random variables that lie in the finite ranges $[\underline\mu_j, \bar\mu_j]$ and $[\underline\omega_j, \bar\omega_j]$, respectively. 
\begin{proposition}\label{proposition1}
The probability that the initial fog node forms a fog network with $\hat{J}$ neighboring nodes  in an iteration including $N$ observations is at least $p'(\gamma) = \sum_{k=\hat{J}}^N  {N\choose k} {p'_s}^k (1-p'_s)^{N-k}$ where $p'_s = 
F_{d_{in}} \scriptstyle{\left( \!\left[ \frac{W_l N_0}{\beta_1 P_{\textrm{tx},i}} \left( 2^{ \left(\frac{1}{\gamma} (\bar\mu_{ij}(\underline x_{ij})-\hat{\lambda}_{ij})+ \hat{\lambda}_{ij} \right) \frac{K}{W_l} } - 1 \right)  \right]^{\frac{-1}{\beta_2}} \right)} $ $
\left(1 - F_{\mu_n} \!\! \left( \frac{1}{\gamma}(\bar{\mu}_j - \hat{\lambda}_{ij}) + \hat{\lambda}_{ij} \right) \right) 
F_{\omega_n} \left( \gamma \underline\omega_j \right)$.
\end{proposition}
\begin{proof} See Appendix~\ref{proof_proposition1}. 
\end{proof}
\noindent By using the probability in Proposition~\ref{proposition1}, the first condition of Theorem~\ref{theorem3} can be replaced with the condition that $p'(\gamma)$ is very close to 1. This is due to the fact that, for a given $\gamma$, a fog network is always formed with $\hat{J}$ neighboring nodes if $p'(\gamma) = 1$. We define $\bar\gamma_s$ as the smallest value of $\gamma$ with which the initial fog node forms a network including $\hat{J}$ neighboring nodes with probability $p'(\gamma) = 1$ in an iteration including $N$ observations, i.e., $\bar\gamma_s = \min(\{\gamma | p'(\gamma)  =1\})$. 

\begin{figure}[]
\centering
\includegraphics[width=0.45\textwidth]{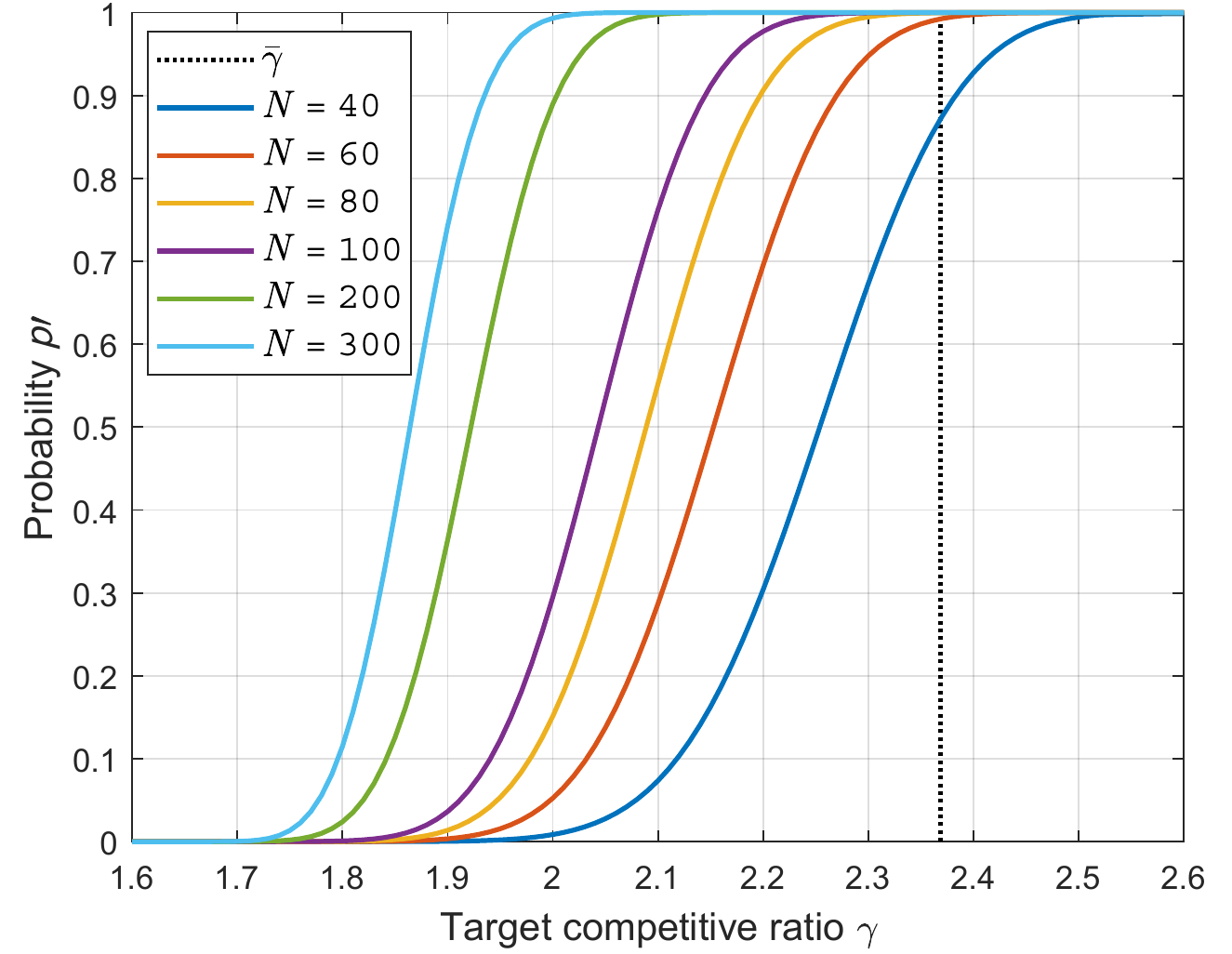}
\caption{{\color{black}Example of the probability $p'$ derived in Proposition~\ref{proposition1}.}}
\label{fig:proposition1}\vspace{-3mm}
\end{figure}

Fig.~\ref{fig:proposition1} shows the upper bound $\bar\gamma$ derived in Theorem~\ref{theorem3}. Fig.~\ref{fig:proposition1} also shows the probability $p'(\gamma)$ derived in Proposition~\ref{proposition1} with respect to the target competitive ratio $\gamma$ for different numbers of observations $N$. In Fig.~\ref{fig:proposition1},  the neighboring nodes are randomly located on a circular area with the maximum distance $\bar d_{ij}=50$~m. Also, $\mu_n$ and $\omega_n$  follow uniform distributions in the ranges $[15,40]$ and $[0.05,0.10]$, respectively. In Fig.~\ref{fig:proposition1}, we use $h=1$, $\hat{J}=6$, $\hat\lambda_{ij}=1.4$, and $l=1$. In Fig.~\ref{fig:proposition1}, if the initial fog node sets $\gamma = \bar{\gamma}_s$, we can see that $p'(\gamma)=1$ for a large value of $N$. For example, the probability $p'(\gamma)$ is one when $\gamma= 2.08$ and $N=300$. In this case, since the first condition of Theorem~\ref{theorem3} is satisfied with a probability close to one, the iteration process for updating $\gamma$ will terminate if the optimal task allocation is achieved. Also, Fig.~\ref{fig:proposition1} shows that $\bar{\gamma}_s$ becomes larger with small $N$. This is due to the fact that the initial fog node must increase $\bar{\gamma}_s$ to select its neighboring nodes within a small number of observations. Since $p'(\gamma)$ approaches to one with increasing $\gamma$, it is possible to determine $\bar\gamma_s$ by numerically finding the smallest $\gamma$ such that $p'(\gamma)$ is very close to 1. Then, in Fig.~\ref{fig:proposition1}, we can observe that $p'(\bar\gamma_s)$ becomes one. Consequently, by setting the initial value of the target competitive ratio $\gamma$ to $\bar{\gamma}_s$, the results of Proposition~\ref{proposition1} can be used to prevent any trial and error in the network formation stage. If the conditions of Theorem~\ref{theorem3} are satisfied, a network can be formed at once, and updating $\gamma$ is not required. To do so, the initial fog node however has to know the information assumed to derive $\bar{\gamma}_s$. When the information is unknown, the proposed framework in Fig.~\ref{fig:flowchart} can be used to iteratively optimize the target competitive ratio. 

\vspace{-3mm}
\section{Simulation Results and Analysis}\label{sec:simulations}

For our simulations, we use a MATLAB simulator\footnote{For further validation of our results, future works can implement the system on an actual fog networking testbed.} in which we consider an initial fog node that can connect to neighboring fog nodes uniformly distributed within a circular area of radius $50~\text{m}$. The arrival sequence of the  fog nodes follows a uniform distribution. The task arrival rate at fog node $i$ is $x_i=10$ packets per second. The computing service rate of the fog nodes is randomly drawn from a uniform distribution over a range of $15$ to $40$ packets per second. All statistical results are averaged over a large number of simulation runs. Similar to prior work \cite{lee2017online}, the simulation results are evaluated with the parameters listed in Table~\ref{table:para}. 

\begin{table}[t]
\centering
\caption{  Simulation parameters}
\label{table:para}
\begin{tabular}{|c|c|}
\hline\centering
\footnotesize Notation & Value   \\ \hline
\footnotesize  $\omega_i\!=\!\omega_j$, $\omega_c$ & $50$,  $25$ msec/packet\\ \hline 
\footnotesize  $\underline\mu_{i} =\underline\mu_{j}$, $\bar{\mu}_i=\bar{\mu}_j$ & $15$,  $40$ packet/sec \\ \hline
\footnotesize  $N$, $\tau$ & $300$, $0.002$ ($0.005$ in Fig.~\ref{fig:fig_main5_gamma}) \\ \hline
\footnotesize  $P_{\textrm{tx},i}$, $\beta_1$, $\beta_2$, {\color{blue}$h$} &$20$~dBm, $10^{-3}$, $4$, {\color{blue}1} \\ \hline
\footnotesize  $K$ &  $64$~kilobytes \\ \hline
\footnotesize  $B$, $N_0$ &  $3$~MHz, $-174$~dBm/Hz \\ \hline
\end{tabular}\vspace{-3mm}
\end{table}

\vspace{-3mm}
\subsection{Performance Evaluation of the Online Optimization Framework}

\begin{figure}[]
\centering
\includegraphics[width=0.45\textwidth]{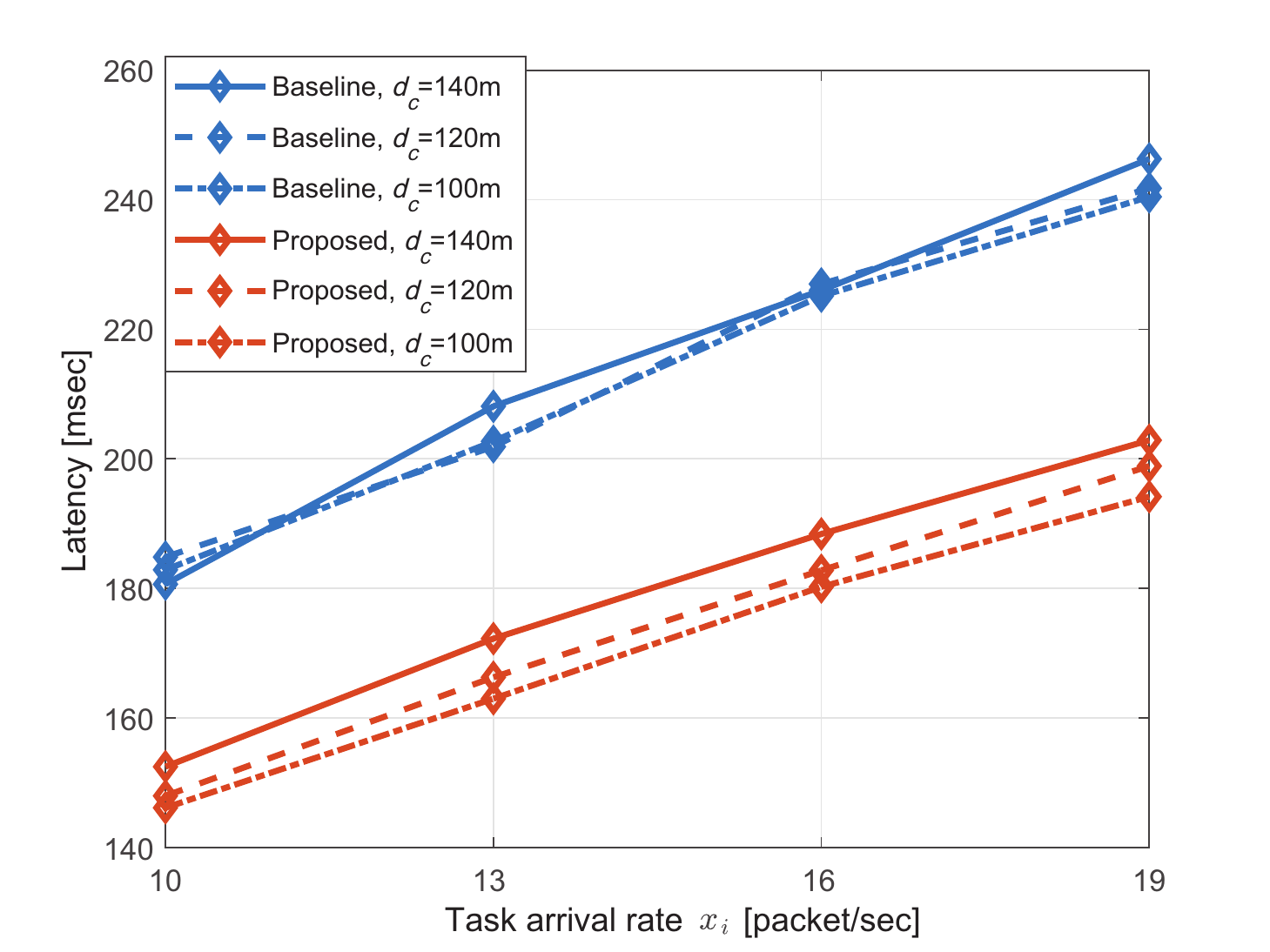}
\caption{Latency for different task arrival rates at the initial fog node $i$.}
\label{fig:fig_main5_mode1_latency}\vspace{-3mm}
\end{figure}

Fig.~\ref{fig:fig_main5_mode1_latency} shows the latency when the total task arrival rate increases from $10$ to $19$ packets per second with $d_c=100$, $120$, and $140$~m, respectively. For comparison purposes, we use a baseline algorithm in which the algorithm observes the first $110$ over $300$ observations nodes and then selects the neighboring nodes from the rest of the arrivals by using the \emph{secretary algorithm} in \cite{lee2017online}. In Fig.~\ref{fig:fig_main5_mode1_latency}, we show that the proposed  framework can reduce the latency  compared to the baseline, for all task arrival rates. For instance, the latency can be reduced by up to $19.25$\% compared to the baseline when $x_i=19$ and $d_c=140$~m. Also, from Fig.~\ref{fig:fig_main5_mode1_latency}, we can see that the latency decreases as the distance to the cloud is reduced. With a shorter distance to the cloud, the cloud transmission service rate becomes higher. Therefore, the cloud is able to process more tasks with a low latency, and the overall latency of the fog network is improved. For example, at $x_i=19$, if $d_c$ decreases from $140$~m to $100$~m, the latency is reduced by $4.29\%$. Moreover, we show that the latency decreases as less tasks arrive at the initial fog node $i$. For instance, when $x_i$ decreases from $19$ to $10$, the latency is reduced by about $25\%$ with $d_c=100$~m. 

\begin{figure}[]
	\centering
	\includegraphics[width=0.45\textwidth]{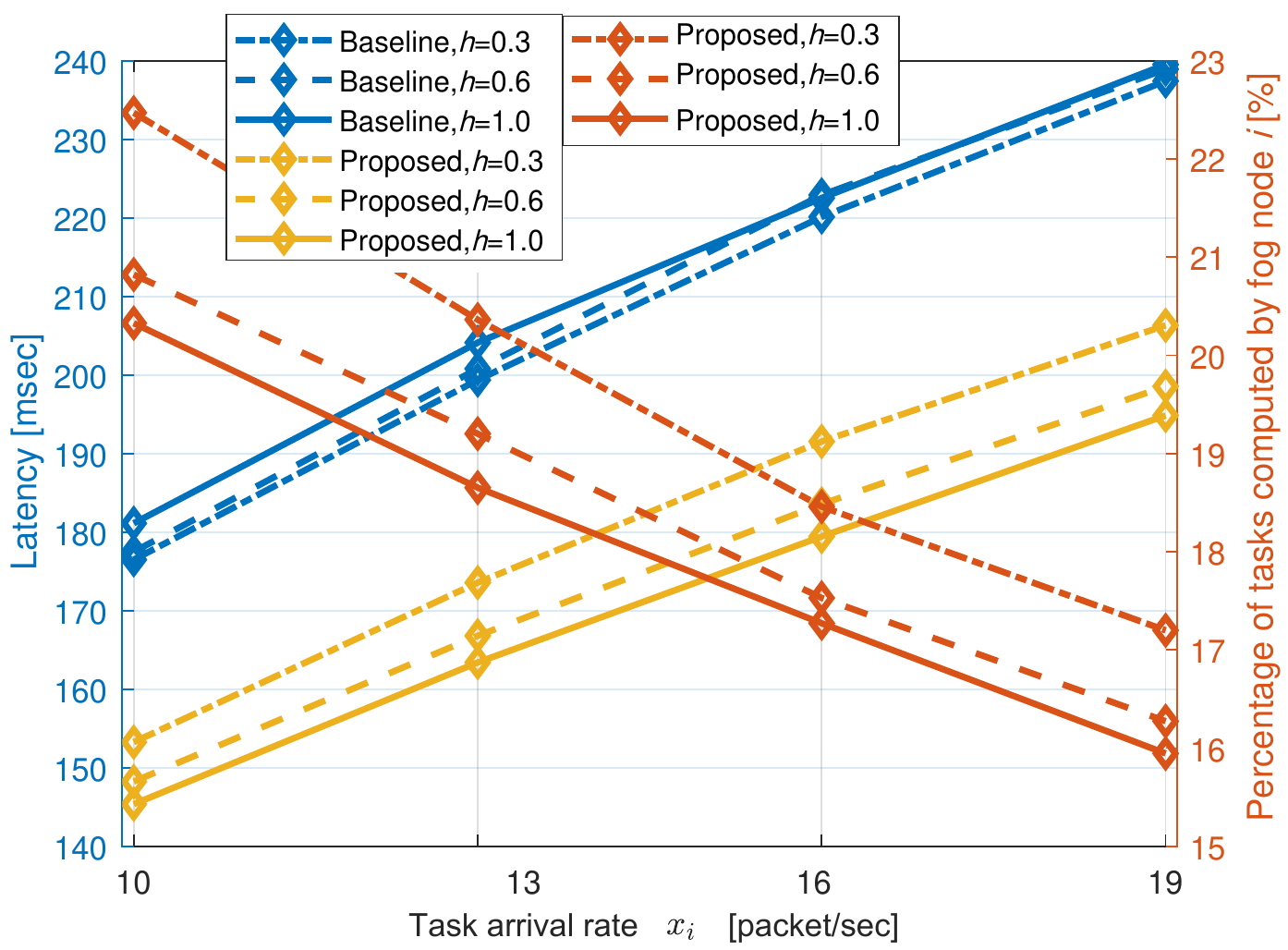}
	\caption{Computing latency and percentage of tasks processed at the initial fog node $i$.}
	\label{fig:fig_revised}\vspace{-3mm}
\end{figure}

Fig.~\ref{fig:fig_revised} shows the latency and the percentage of tasks processed at the initial fog node $i$ when the total task arrival rate increases from $10$ to $19$ packets per second with average fading gain values of $h=0.3$, $0.6$, and $1.0$, respectively. In Fig.~\ref{fig:fig_revised}, we show that the latency decreases as the average fading gain increases, for all task arrival rates. For a higher average fading gain, the transmission service rates of the fog computing nodes become larger. Therefore, the tasks can be efficiently offloaded, with low latency, to neighboring fog nodes and the cloud hence improving the overall latency of the fog network. Also, from Fig.~\ref{fig:fig_revised}, we can see that the percentage of tasks processed at the initial fog node $i$ decreases as the total task arrival rate $x_i$ increases. Moreover, Fig.~\ref{fig:fig_revised} shows that the initial fog node $i$ tends to process more tasks when $h$ is smaller. This is due to the fact that a smaller $h$ increases the wireless transmission latency required to offload tasks to other computing nodes. For example,  at $x_i=10$, if $h$ increases from $0.3$ to $1.0$, the percentage of tasks processed at node $i$ increases by up to about $10\%$. 

\begin{figure}[]
\centering
\includegraphics[width=0.45\textwidth]{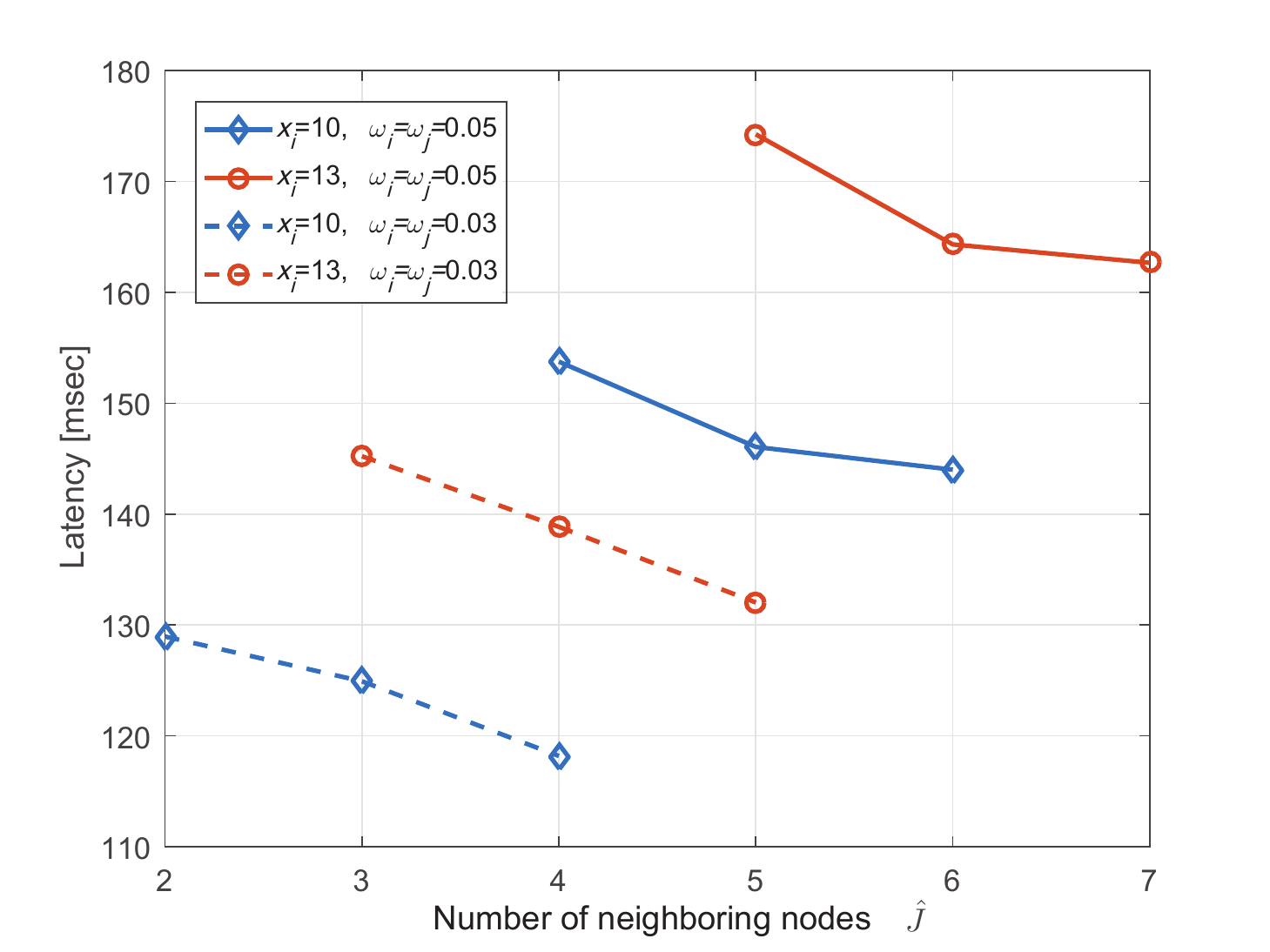}\vspace{-2mm}
\caption{Latency for different number of neighboring nodes.}
\label{fig:fig_main5_mode12_latency}\vspace{-3mm}
\end{figure}

Fig.~\ref{fig:fig_main5_mode12_latency} shows the relationship between the latency and the number of neighboring nodes when the total task arrival rate is given by $x_i=10$ and~$13$ packets per second, respectively, and the processing delays of the fog nodes are given by $\omega_i=\omega_j=50$ and $30$ milliseconds, respectively. In Fig.~\ref{fig:fig_main5_mode12_latency}, a smaller processing delay indicates that the fog nodes have a higher processing speed. From Fig.~\ref{fig:fig_main5_mode12_latency}, we can observe the tradeoff between scenarios having a large number of fog nodes with low processing power and scenarios having a small number of fog nodes with high processing power. If fog nodes with higher processing speed are deployed,  latency is reduced, and the formed network size decreases. This is due to the fact that the fog nodes having a faster processing speed do not need to form a large network. In fact, a larger network size can lead to lower transmission service rates. For instance, if the processing delay of fog nodes decreases from $50$ to $30$ milliseconds, the latency is reduced by up to $18.8\%$ while the number of neighboring nodes decreases from 7 to 5.

\begin{figure}[]
\centering
\includegraphics[width=0.45\textwidth]{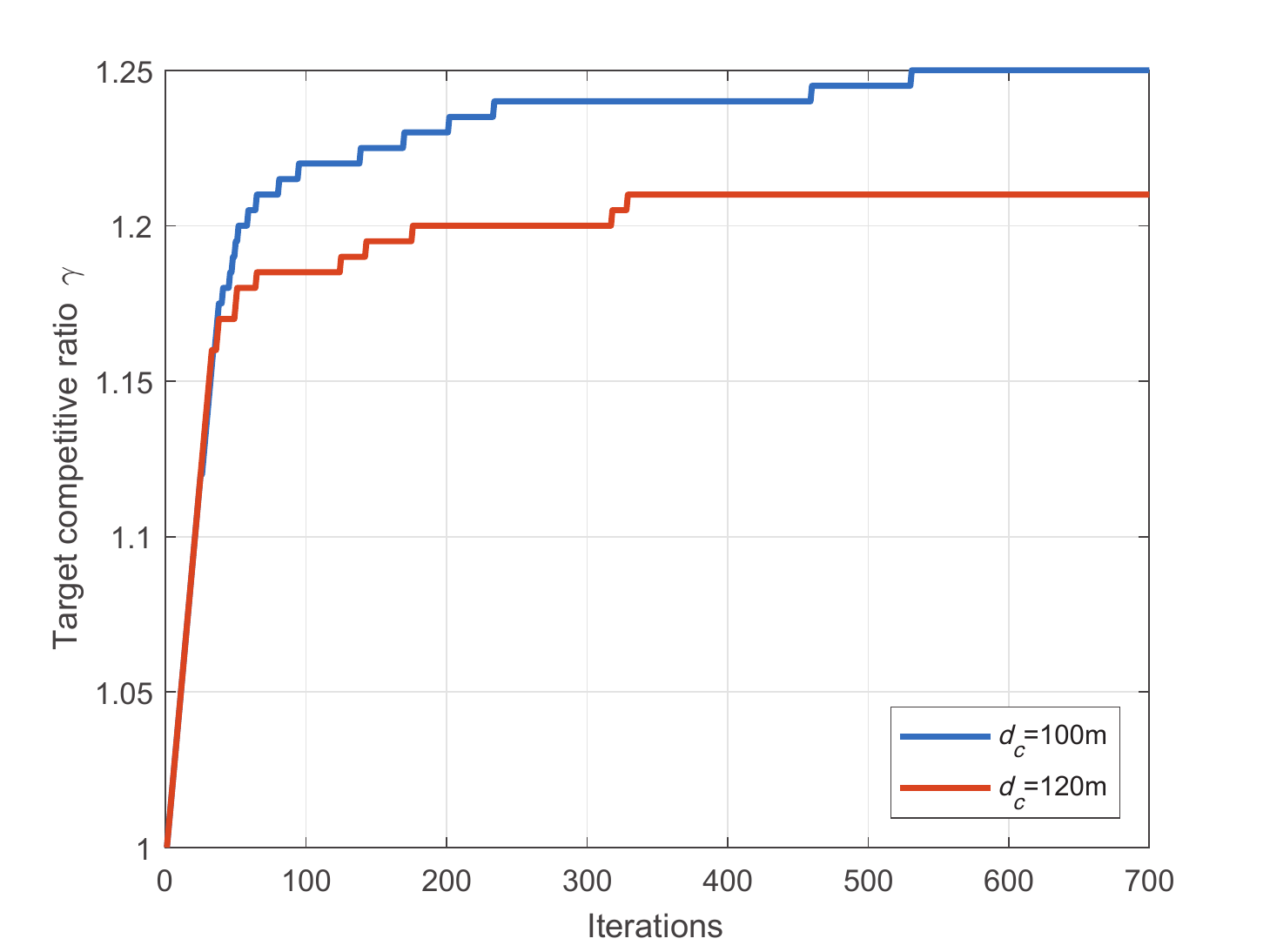}\vspace{-2mm}
\caption{Changes in the target competitive ratio $\gamma$ over 700 updates.}
\label{fig:fig_main5_gamma}\vspace{-3mm}
\end{figure}

Fig.~\ref{fig:fig_main5_gamma} plots the value of $\gamma$ during 700 updates for different distances to the cloud, $d_c=100$~m and $120$~m, respectively. Fig.~\ref{fig:fig_main5_gamma} shows that the value of $\gamma$ approaches a constant value. For instance, $\gamma$  first reaches 1.17 at 38 iterations with $d_c=120$~m. Then, $\gamma$ becomes 1.21 at 329 iterations, and this value is maintained thereafter. From Fig.~\ref{fig:fig_main5_gamma}, we can see that fog node $i$ can find a proper $\gamma$ after a finite number of trials and updates. Also, the results of Fig.~\ref{fig:fig_main5_gamma} show that $\gamma$ becomes larger as the distance to the cloud is closer. This is because  $\hat{u}$ and the threshold value  decrease when $d_c$ is reduced. If the threshold value decreases, it becomes more challenging to select the $\hat{J}$ neighboring nodes within the limited number of observations since the selected neighboring nodes must have a lower latency than the threshold. Therefore, in order to maintain a proper threshold value, $\gamma$ will be larger when $d_c$ decreases. 

\begin{figure}[]
\centering
\includegraphics[width=0.45\textwidth]{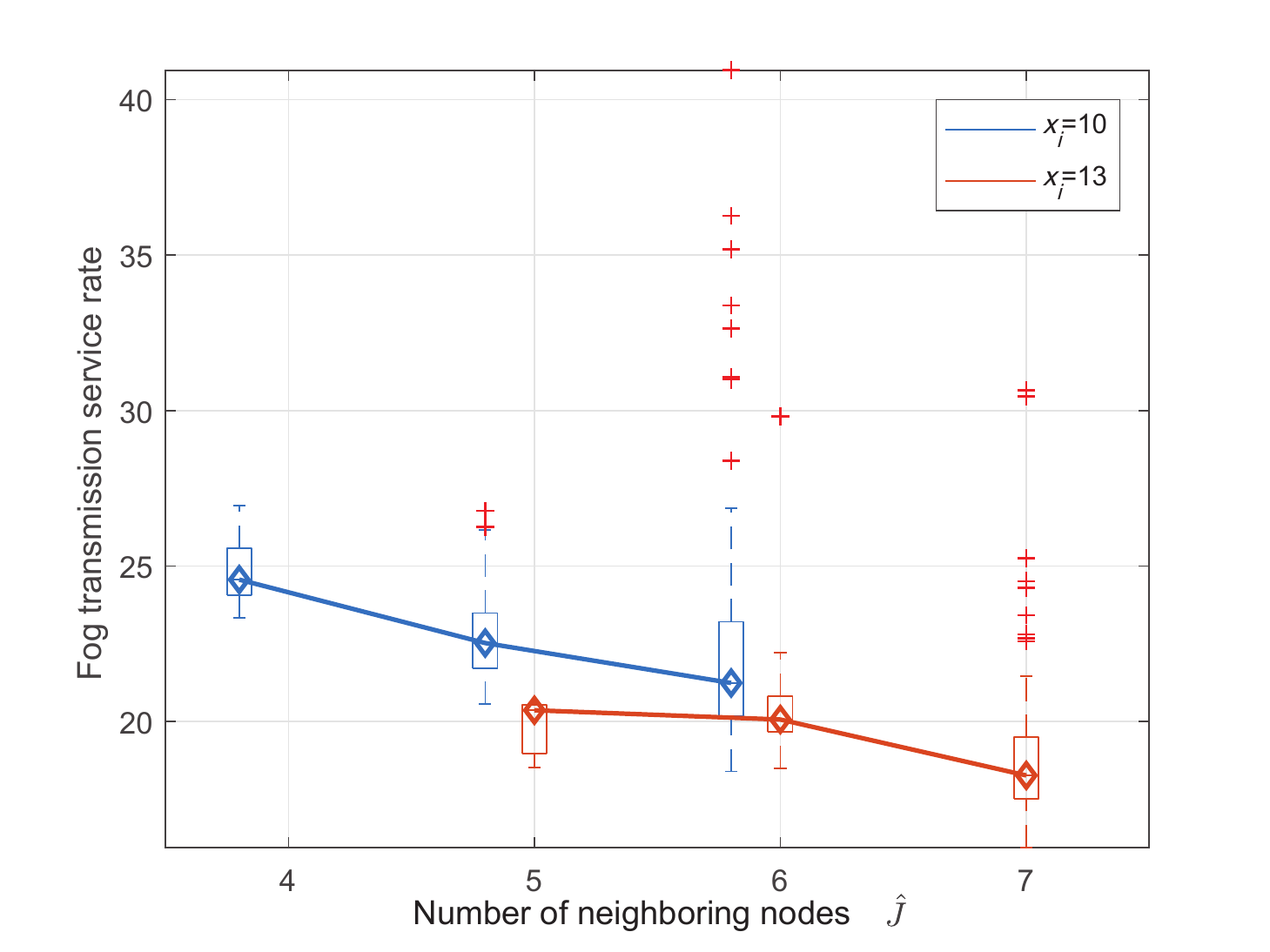}\vspace{-2mm}
\caption{Fog transmission service rate with respect to the number of neighboring nodes. }
\label{fig:fig_main5_mode11_fogtxrate}\vspace{-3mm}
\end{figure}

Fig.~\ref{fig:fig_main5_mode11_fogtxrate} shows the relationship between the fog transmission service rate and the number of neighboring nodes  when $x_i=10$ and~$13$, respectively. Here, we can see that the fog transmission service rate increases as the number of neighboring nodes decreases. This stems from the fact that the bandwidth per node increases as less fog nodes share the total bandwidth. For instance, the fog transmission service rate can increase by $15.6\%$ if $\hat{J}$ goes from 6 to 4 with $x_i=10$. Fig.~\ref{fig:fig_main5_mode11_fogtxrate} also shows that the formed network size becomes larger if $x_i$ increases. This is due to the fact that offloading tasks to a larger size of the network can reduce the tasks per node, and, hence, the maximum latency of the network will decrease. For instance, when $x_i=10$, the range of $\hat{J}$ is between 4 and 6. However, if $x_i= 13$,  $\hat{J}$ falls in the range between 5 and 7. 

\begin{figure}[]
\centering
\includegraphics[width=0.45\textwidth]{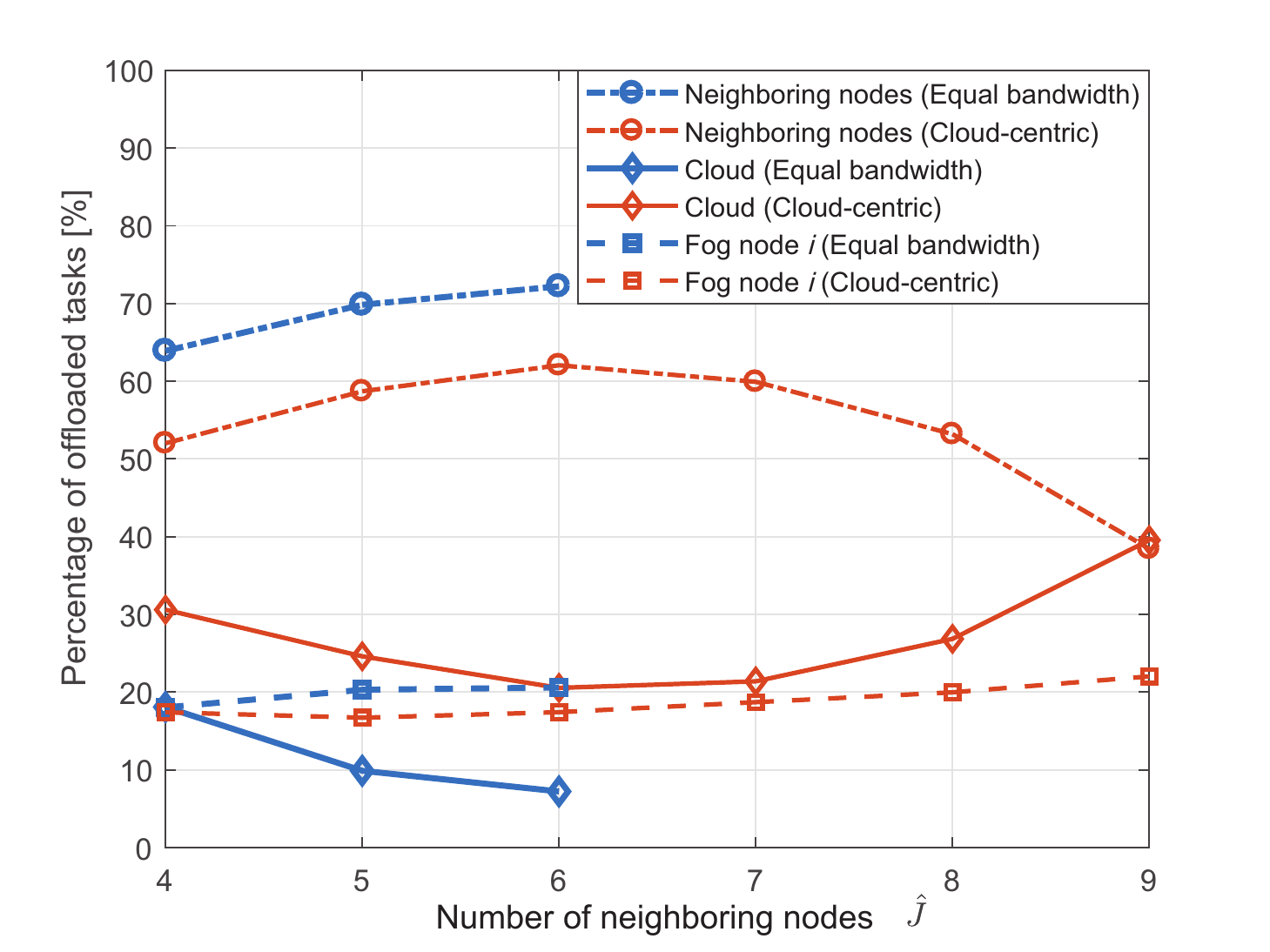}\vspace{-2mm}
\caption{Task distribution with respect to the number of neighboring nodes.}
\label{fig:fig_main5_mode2_task}\vspace{-3mm}
\end{figure}

In Fig.~\ref{fig:fig_main5_mode2_task}, we show the task distribution among neighboring nodes, the cloud, and fog node~$i$ for different numbers of neighboring nodes when two bandwidth allocation approaches are used, respectively. It can be seen that the cloud-centric bandwidth allocation increases the  tasks offloaded to the cloud  when compared to the equal-bandwidth allocation. This is because the cloud transmission service rate increases, so offloading more tasks to the cloud can lower the latency. For instance, if the cloud-centric bandwidth allocation is used and $\hat{J}=4$, the cloud is allocated $22.86$\% more tasks than in the case of equal bandwidth allocation. Also, in Fig.~\ref{fig:fig_main5_mode2_task}, we show that the optimal network size is different, depending on the bandwidth allocation scheme. For instance, the cloud-centric bandwidth allocation yields a larger network size   than the equal bandwidth allocation. When the network size is large, the cloud can maintain a high transmission service rate by using the cloud-centric bandwidth allocation. Therefore, the high cloud transmission service rate enables to offload most tasks to the cloud with a low transmission latency. For example, Fig.~\ref{fig:fig_main5_mode2_task} shows that the number of neighboring nodes is between 4 and 6 if equal bandwidth allocation is used. However, if the cloud-centric bandwidth allocation is used, the number of neighboring nodes varies from 4 to 9. Moreover, Fig.~\ref{fig:fig_main5_mode2_task} shows that the number of  tasks offloaded to the cloud decreases when $\hat{J}$ increases from 4 to 6 for both bandwidth allocation schemes. In this phase, the number of  tasks offloaded to neighboring nodes will increase because offloading more tasks at the fog layer can reduce the latency at the cloud. However, if the number of neighboring nodes increases when using the cloud-centric bandwidth allocation, e.g., there are 7 or more neighboring nodes, the number of tasks offloaded to the neighboring nodes will decrease with the network size. This is due to the fact that the fog transmission service rates are smaller for larger networks which yields higher fog transmission latency. As a result, more tasks will be allocated to the cloud so as to utilize its fast computing resources. 

\vspace{-3mm}
\subsection{Performance Evaluation of Algorithm~\ref{algorithm} for a fixed $\gamma$} 

\begin{figure}[]
\centering
\includegraphics[width=0.45\textwidth]{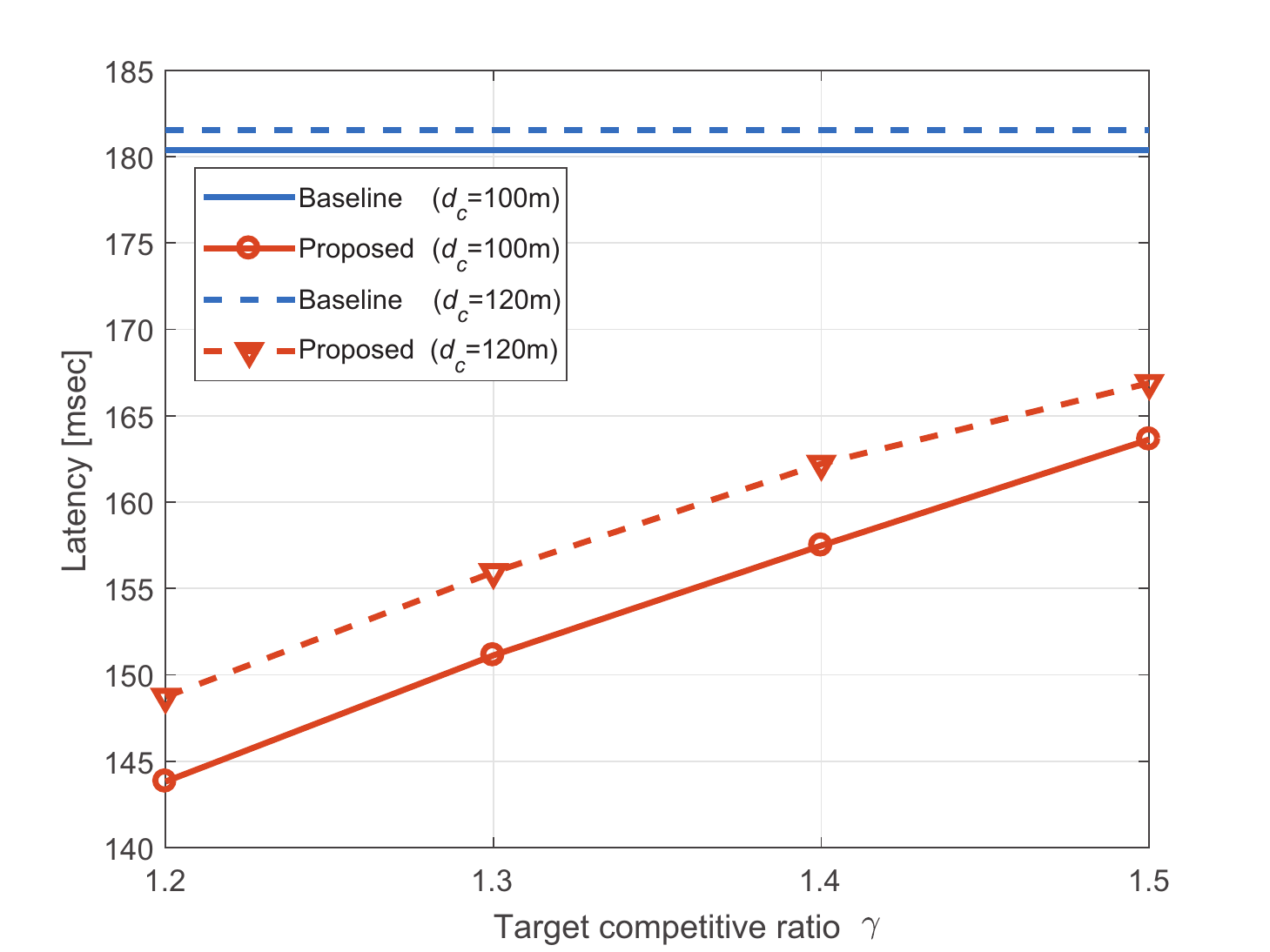}\vspace{-2mm}
\caption{Latency comparison versus the target competitive ratio.}
\label{fig:fig_main2_mode0_latency}\vspace{-3mm}
\end{figure}

In Figs.~\ref{fig:fig_main2_mode0_latency}~and~\ref{fig:fig_main2_observations}, we evaluate the performance of Algorithm~\ref{algorithm} when the proposed framework uses a fixed value of $\gamma$ without constraint \eqref{problem1:c3}. While the target competitive ratio is used in the proposed framework to determine the threshold value and make a decision on node selection, the baseline algorithm has a different mechanism to determine threshold values. Therefore, the latency results of the baseline do not depend on the target competitive ratio. By using a predefined $\gamma$, the update step of $\gamma$ is not needed, which can be useful for scenarios in which the delay of this update can hinder the network latency. Fig.~\ref{fig:fig_main2_mode0_latency} shows the latency for the different preset values of $\gamma$ ranging from 1.2 to 1.5 with $d_c=100$~m and $120$~m, respectively. From Fig.~\ref{fig:fig_main2_mode0_latency}, we can see that the proposed framework achieves lower latencies than the baseline, for all $\gamma$. For instance, the latency of the proposed framework can be reduced by up to $20.3\%$ compared to that of the baseline if $\gamma=1.2$ and $d_c=100$~m. Also, Fig.~\ref{fig:fig_main2_mode0_latency} shows that the latency achieved by the proposed framework becomes smaller when $\gamma$ decreases. This stems from the fact that a low threshold value with small $\gamma$ allows the initial fog node to only select neighboring nodes having a high performance. For example, the latency can be reduced by up to $12.1$\% if $\gamma$ decreases from 1.5 to 1.2 with $d_c=100$~m. 

\begin{figure}[]
\centering
\includegraphics[width=0.45\textwidth]{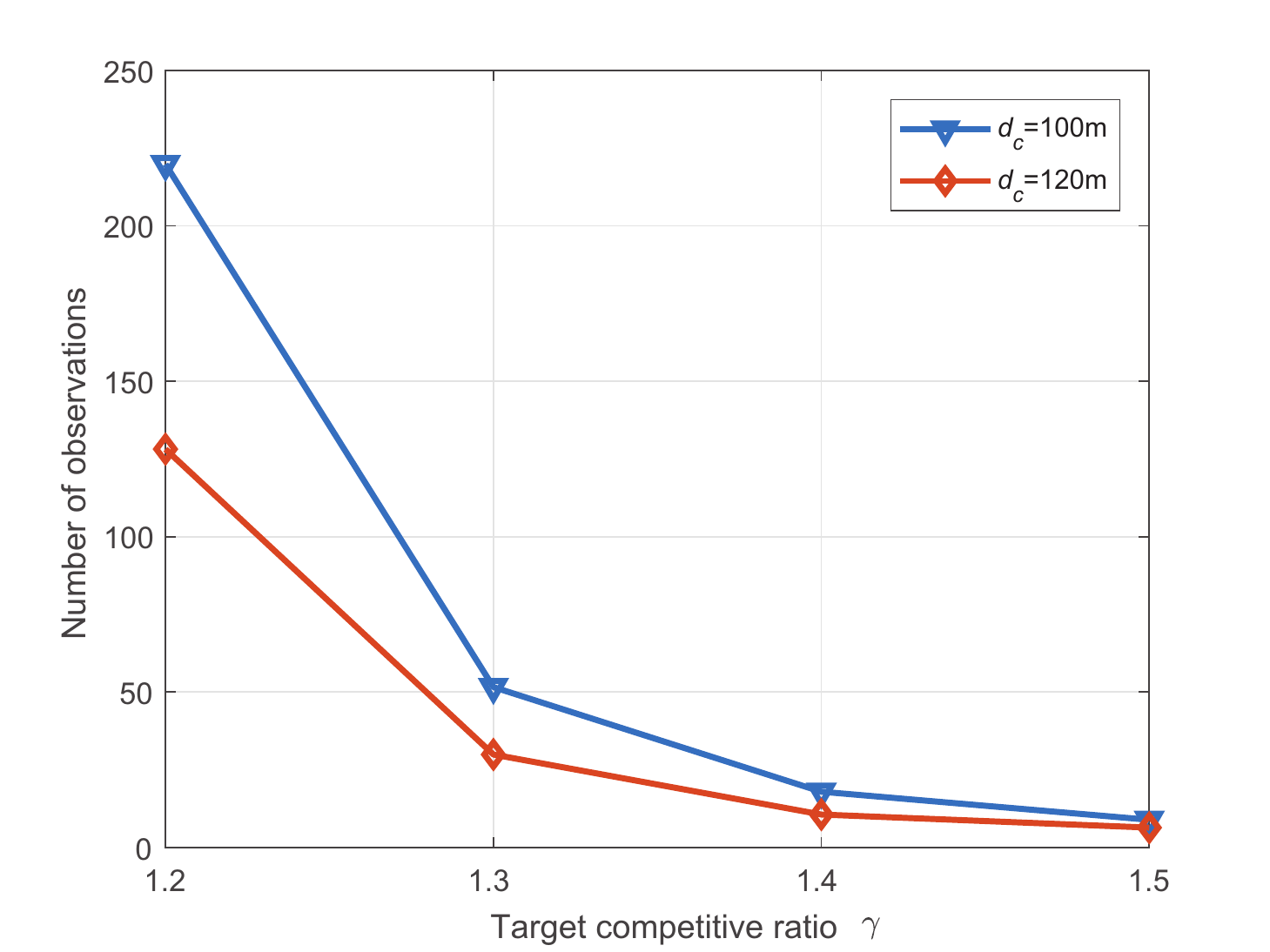}
\caption{The required number of observations for different values of $\gamma$. }
\label{fig:fig_main2_observations}\vspace{-3mm}
\end{figure}

In Fig.~\ref{fig:fig_main2_observations}, we show the number of observations of the neighboring node arrivals until $\hat{J}$ neighboring  nodes are selected for different $\gamma$ with $d_c=100$~m and $120$~m, respectively. In this figure, we can see that a large value of $\gamma$ results in a small number of observations due to the associated increase in the threshold value. For instance, as $\gamma$ increases from $1.2$ to $1.5$, the number of observations can be reduced by about $96\%$ with $d_c=100$~m. Fig.~\ref{fig:fig_main2_observations} shows that a large value of $d_c$ lowers the number of observations since increasing $d_c$ results in a large $\hat{u}$ and threshold value. For example, the number of observations can be reduced by about $42$\% if $d_c$ increases from 100~m to 120~m with $\gamma=1.2$. Moreover, from Figs.~\ref{fig:fig_main2_mode0_latency}~and~\ref{fig:fig_main2_observations}, we can characterize the tradeoff between the latency  and the number of observations. In particular, a small $\gamma$ results in a lower latency, but requires a large number of observation.

\begin{figure}[]
\centering
\includegraphics[width=0.45\textwidth]{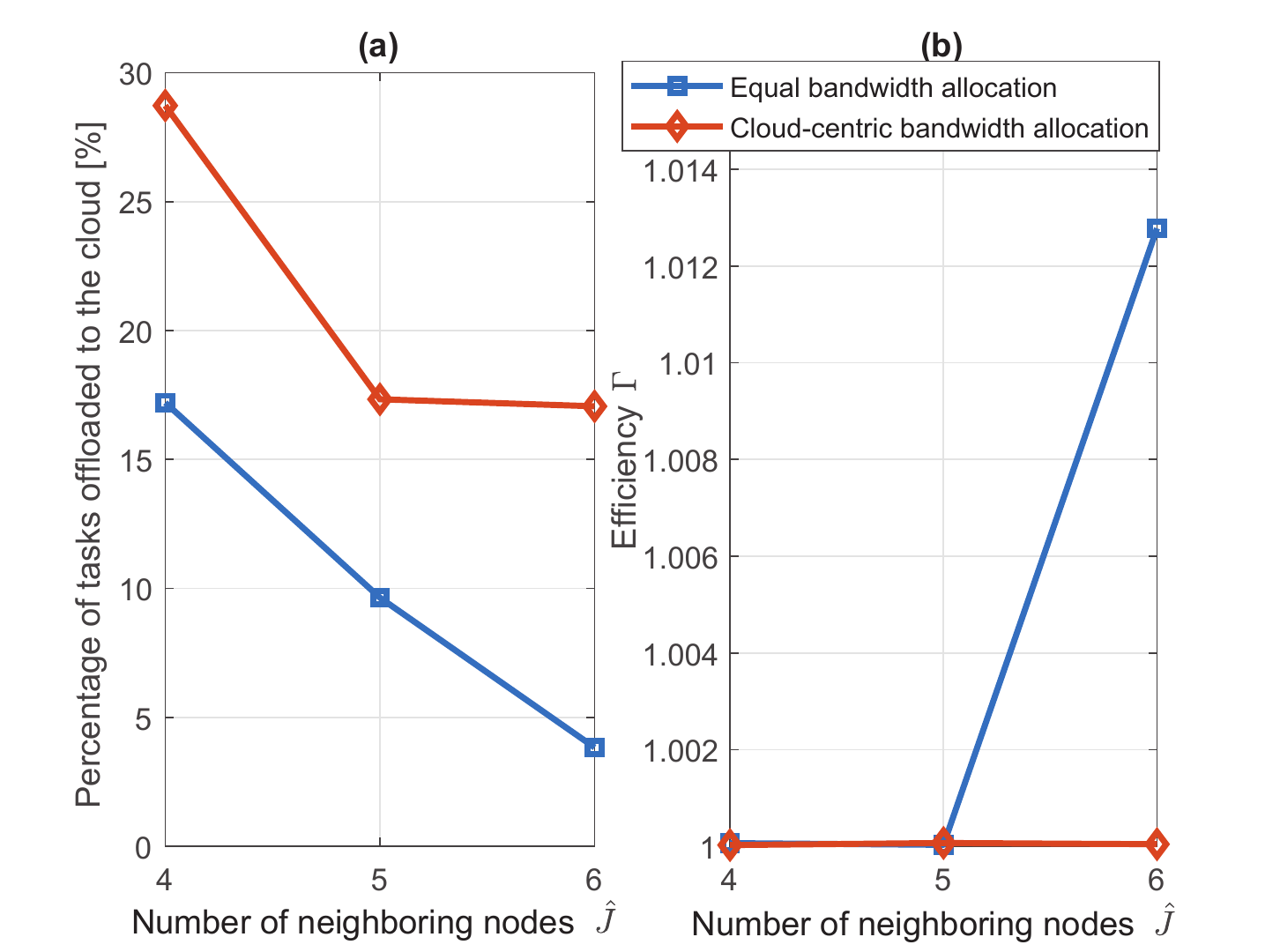}
\caption{Performance comparison of two bandwidth allocation schemes with respect to the number of neighboring nodes. }
\label{fig:fig_main2_mode1_merge}\vspace{-3mm}
\end{figure}

Fig.~\ref{fig:fig_main2_mode1_merge} shows the percentage of tasks offloaded to the cloud and the scheduling efficiency of the task distribution when two bandwidth allocation schemes are used, respectively, with $\gamma=1.2$ and $d_c=100$~m. In Fig.~\ref{fig:fig_main2_mode1_merge}~(a), the tasks offloaded to the the cloud decreases as the number of fog nodes increases since the cloud transmission service rate decreases. Also, Fig.~\ref{fig:fig_main2_mode1_merge}~(b) shows that, when  equal bandwidth allocation is used for a large network size, the scheduling efficiency may not be optimal, i.e., $\Gamma > 1$ due to a large  latency for the transmissions to  the cloud. In this case, though the equal-bandwidth allocation still achieves $\Gamma$ that is close to 1, the cloud-centric bandwidth allocation can be used to enhance efficiency. This is because the cloud-centric bandwidth allocation increases the  cloud transmission service rate by allocating more bandwidth. It can be seen for instance that the equal bandwidth allocation yields $\Gamma=1.013$ in the case of $6$ neighboring nodes, but the efficiency of the cloud-centric bandwidth allocation becomes $\Gamma=1$. 

\vspace{-3mm}
\subsection{Optimal Network Size in an Offline Setting}

\begin{figure}[]
\centering
\includegraphics[width=0.45\textwidth]{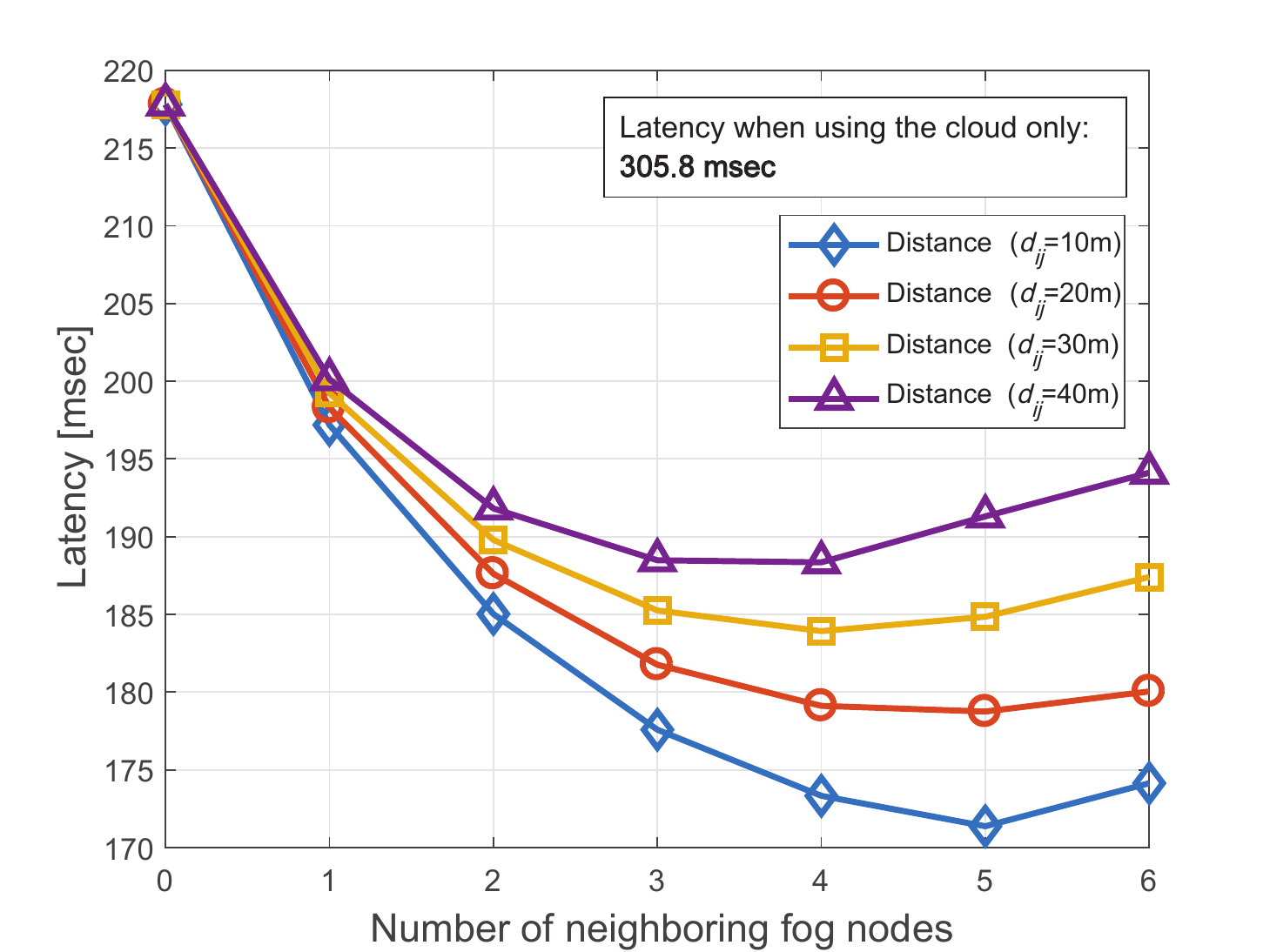}
\caption{Latency for different number of neighboring fog nodes in an offline setting.}
\label{fig:main1_latency}\vspace{-3mm}
\end{figure}

Fig.~\ref{fig:main1_latency} shows  the optimal latency for different network sizes when all neighboring nodes are located at $d_{ij}$ varying from $10$~m to $40$~m. In Fig.~\ref{fig:main1_latency}, it is assumed that complete information on the network is known and that the fog nodes have  identical parameters, i.e., $\mu_i=\mu_j=20$ when $d_c=150$~m. In this offline setting, we study the impact of the network size on the latency  by using an offline optimization solver to find the optimal latency for a given network.  Fig.~\ref{fig:main1_latency}  shows that the optimal latency is directly affected by the number of neighboring nodes. When the network size increases, latency starts to decrease since fewer tasks can be offloaded to each neighboring node. However, if the network size increases, the latency will eventually increase since the bandwidth per node is smaller. For example, the optimal latency decreases when the number of neighboring nodes increases from 1 to 3 with $d_{ij}=40$. However, once the number of neighboring nodes increases beyond 3, the latency starts to increase. Moreover, from Fig.~\ref{fig:main1_latency}, we can see that the optimal network size changes with the distances between fog nodes. For instance, for $d_{ij}=40$~m, the latency can be minimized when there are $3$ neighboring nodes in the fog network. However, if $d_{ij}=10$~m, the latency is minimized when the number of neighboring nodes is $5$. Therefore, if the fog transmission service rate is high (for shorter distances), increasing the number of neighboring nodes to 5 can reduce the latency. On the other hand, if the fog transmission service rate is low (due to poor wireless channel), having a smaller network size with $3$ nodes is required to minimize the latency. Also, we note that the results in Fig.~\ref{fig:main1_latency} show that there exists an optimal network size that can be found by running Phase 1 of Algorithm~\ref{algorithm}. Finally, Fig.~\ref{fig:main1_latency} clearly shows that the latency is reduced by offloading the tasks to both the fog layer and the cloud, instead of relying solely on the cloud. For example, if the tasks are offloaded to the cloud, initial fog node, and 5 neighboring nodes located at $d_{ij}=10$~m, the latency can be reduced by up to $43.9\%$ compared to the case using the cloud only. 

\vspace{-3mm}
\section{Conclusion and Future Work}\label{sec:conclusion}

In this paper, we have proposed a novel framework to jointly optimize the formation of fog networks and the distribution of computational tasks in a hybrid fog-cloud system. We have addressed the problem using an online optimization formulation whose goal is to minimize the maximum latency of the nodes in the fog  network in  presence of uncertainty about fog nodes' arrivals. To solve the problem, we have proposed online optimization algorithms whose target competitive ratio is achieved by suitably selecting the neighboring nodes while effectively offloading the tasks to the neighboring fog nodes and the cloud. The theoretical analysis and simulation results have shown that the proposed framework achieves a low target competitive ratio while successfully minimizing the maximum latency in fog computing. Extensive simulation results are used to showcase the performance benefits of the proposed approach. For future work, a dynamic bandwidth scheme can be designed to further reduce the latency. Also, packet prioritizing can be adopted at the initial fog node to meet different service-level latency requirements. Moreover, the proposed framework can be extended to the scenario in which multiple fog networks are formed by multiple initial fog nodes. Further, the proposed fog network formation algorithm can be extended to account for the instantaneous fading by using advanced techniques such as stochastic optimization. Finally, one important future work is to conduct an experimental analysis pertaining to fog computing over an actual wireless testbed.

\appendices
\vspace{-3mm}
\section{Proof of Proposition~\ref{proposition1}}\label{proof_proposition1}
  
\begin{proof}
For a given $\gamma$, the arriving node $n$ is selected by the initial fog node if $D_n(\hat{\lambda}_{ij}, \mu_{in}) \leq \gamma \hat{u}$. 
The probability of node selection event $E_s$ is $p_s = \textrm{Pr} \left\{  D_n(\hat{\lambda}_{ij}, \mu_{in}) \leq \gamma \hat{u} \right\}$. 
With the same target competitive ratio $\gamma$, $E$ is defined as the event where $E_s$ happens more than $\hat{J}$ times during $N$ trials within an iteration. Since event $E$ is a sufficient condition to form a network for a given $\gamma$, the probability to form a network is at least given by $p = \sum_{k=\hat{J}}^N  {N\choose k} p_s^k (1-p_s)^{N-k}$ where $N$ is the maximum number of observations allowed within an iteration, and all inputs $\sigma_n, \forall n \in[1,N]$ are independent. 

Since $
	p_s = \textrm{Pr}\Big\{  \frac{1}{\mu_{in}-\hat{\lambda}_{ij}} + \frac{1}{\mu_{in}} + \frac{1}{\mu_{n}-\hat{\lambda}_{ij}} + \frac{1}{\mu_{n}} + 2 \omega_n
	\leq $ $\gamma \left( \frac{1}{\bar\mu_{ij}-\hat{\lambda}_{ij}} + \frac{1}{\bar\mu_{ij}} + \frac{1}{\bar\mu_j-\hat{\lambda}_{ij}} + \frac{1}{\bar\mu_j} + 2 \underline \omega_j\right)	
	\Big\}$, a lower bound of $p_s$ can be given by $p'_s = \textrm{Pr}\left\{  E_1 \cap  E_{1'} \cap E_2 \cap  E_{2'} \cap E_3 \right\}$ 
where $E_1$ is the event where $\frac{1}{\mu_{in}-\hat{\lambda}_{ij}} - \gamma \frac{1}{\bar\mu_{ij}-\hat{\lambda}_{ij}} \leq 0$, 
$E_{1'}$ is the event where $\frac{1}{\mu_{in}} - \gamma \frac{1}{\bar\mu_{ij}} \leq 0$, 
$E_2$ is the event where $\frac{1}{\mu_{n}-\hat{\lambda}_{ij}} - \gamma  \frac{1}{\bar\mu_j-\hat{\lambda}_{ij}} \leq 0$, 
$E_{2'}$ is the event where $\frac{1}{\mu_{n}} - \gamma \frac{1}{\bar\mu_j} \leq 0$, and 
$E_3$ is the event where $\omega_n - \gamma \underline \omega_j \leq 0$. Then, $p'_s$ can be rewritten as $\textrm{Pr}\{  E_{1'} | E_{1}\} \textrm{Pr}\{E_{1}\} \textrm{Pr}\{  E_{2'} | E_{2}\} \textrm{Pr}\{E_{2}\} \textrm{Pr}\{  E_3 \}$. Then, due to the relationship $\frac{1}{\mu_{in}} \leq \frac{1}{\mu_{in}-\hat{\lambda}_{ij}} \leq \gamma \frac{1}{\bar\mu_{ij}-\hat{\lambda}_{ij}} \leq \gamma \frac{1}{\bar\mu_j}$, if the condition for $E_1$ is satisfied, i.e., $\frac{1}{\mu_{in}-\hat{\lambda}_{ij}} \leq \gamma \frac{1}{\bar\mu_{ij}-\hat{\lambda}_{ij}}$, then it is clear that the condition for $E_{1'}$ is also satisfied, i.e., $\frac{1}{\mu_{in}} \leq \gamma \frac{1}{\bar\mu_j}$. This, in turn, implies $\textrm{Pr}\{  E_{1'} | E_{1}\}=1$. Similarly, if $E_2$ happens, then it always incurs $E_{2'}$, and, thus, $\textrm{Pr}\{  E_{2'} | E_{2}\}=1$. In consequence, $p'_s$ can be simplified as $p'_s 	= \textrm{Pr}\{E_{1}\} \textrm{Pr}\{E_{2}\} \textrm{Pr}\{  E_3 \}$. Note that $\textrm{Pr}\{ E_{1} \} $ can be expressed by using  $d_{in}$ since $\mu_{in}$ is a function of $d_{in}$ in \eqref{eq:mu}. When $F_{d_{in}}$, $F_{\mu_n}$, and $F_{\omega_n}$, respectively, are the cumulative probability functions with respect to $d_{in}$, $\mu_{n}$, and $\omega_n$, $\textrm{Pr}\{E_{1}\}$, $\textrm{Pr}\{E_{2}\}$, and $\textrm{Pr}\{  E_3 \}$ are $  \textrm{Pr}\{E_{1}\}=  F_{d_{in}}\scriptstyle{\left( \left[ \frac{W_l N_0}{\beta_1 P_{\textrm{tx},i} }\left( 2^{\left(\frac{1}{\gamma} (\bar\mu_{ij}(\underline x_{ij})-\hat{\lambda}_{ij})+ \hat{\lambda}_{ij} \right) \frac{K}{W_l} } - 1 \right) \right]^{-1/\beta_2} \right)}$,
$\textrm{Pr}\{E_{2}\} = 1-  F_{\mu_n}\left( \frac{1}{\gamma}(\bar\mu_j - \hat{\lambda}_{ij})+\hat\lambda_{ij} \right)$,
 and  $\textrm{Pr}\{  E_3 \} =  F_{\omega_n} \left( \gamma \underline\omega_j \right)$. 
 
Finally, it is clear that $ p' \triangleq \sum_{k=\hat{J}}^N  {N\choose k} {p'_s}^k (1-p'_s)^{N-k} \leq p$ due to $p'_s \leq p_s$. 
Hence, $p'$ is a lower bound of the probability that a given target competitive ratio is used to form a network without an update. 
\vspace{-3mm}
\end{proof}

\bibliographystyle{IEEEtran}

\end{document}